\documentclass[11pt, a4paper]{article}
\usepackage{bm}
\usepackage{bbm}
\usepackage{amsmath,amssymb}
\usepackage{setspace}
\usepackage{graphicx,caption,varwidth}
\usepackage{tikz}
\usepackage{booktabs,caption}
\usepackage[flushleft]{threeparttable}
\usepackage{subcaption}
\usepackage{multirow}
\usepackage{pdflscape}
\usepackage{rotating}
\usepackage{adjustbox}
\usepackage{amsthm}
\usepackage{float}
\usepackage[title]{appendix}
\usetikzlibrary{automata, positioning, chains}

\DeclareMathOperator{\E}{\mathbb{E}}

\DeclareMathOperator{\X}{\mathbf{X}}
\DeclareMathOperator{\D}{\mathbf{D}}
\DeclareMathOperator{\x}{\mathbf{x}}

\DeclareMathOperator{\Loss}{\mathbb{L}}

\DeclareMathOperator*{\argmax}{arg\,max}
\DeclareMathOperator*{\argmin}{arg\,min}

\newtheorem{theorem}{Theorem}

\newcommand{\R}{\textnormal{\sffamily\bfseries R}}
\newcommand{\pkg}[1]{{\fontseries{b}\selectfont #1}} 
\newcommand\independent{\protect\mathpalette{\protect\independenT}{\perp}}
\def\independenT#1#2{\mathrel{\rlap{$#1#2$}\mkern2mu{#1#2}}}

\usepackage{fullpage}

\usepackage[utf8]{inputenc} 
\usepackage[T1]{fontenc}
\usepackage{lmodern}

\usepackage[backend=biber, style=apa, labeldate=year,doi=false,eprint=false]{biblatex} 
\bibliography{BOTR_references}
\DeclareLanguageMapping{english}{english-apa}

\usepackage{authblk}
	
\title{Estimating Bayesian Optimal Treatment Regimes \\ for Dichotomous Outcomes using Observational Data\thanks{The authors gratefully acknowledge the contribution of oropharynx cancer data by I. Nauta, MD, C.R. Leemans, MD PhD, and R.H. Brakenhoff, PhD (Amsterdam UMC, Vrije Universiteit Amsterdam, Otolaryngology / head and neck surgery, Tumor Biology section, Cancer Center Amsterdam, The Netherlands). This research has received funding by the European Commission under Horizon 2020 Framework Programme (Societal Challenges, Health), grant agreement no. 689715, project BD2Decide.}}

\author[1]{Thomas~Klausch}
\author[1]{Peter~M.~van~de~Ven}
\author[1]{Tim~van~de~Brug}
\author[1,2]{Mark~A.~van~de~Wiel}
\author[1]{Johannes~Berkhof}

\affil[1]{\footnotesize Amsterdam UMC, Vrije~Universiteit~Amsterdam, Department~of~Epidemiology~and~Biostatistics, Amsterdam~Public~Health~Research~Institute,~Amsterdam,~The~Netherlands}
\affil[2]{\footnotesize MRC Biostatistics Unit, Cambridge University, Cambridge, UK}

\date{ Email: t.klausch@vumc.nl}

\begin{document}
\maketitle
\begin{abstract}
Optimal treatment regimes (OTR) are individualised treatment assignment strategies that identify a medical treatment as optimal given all background information available on the individual. We discuss Bayes optimal treatment regimes estimated using a loss function defined on the bivariate distribution of dichotomous potential outcomes. The proposed approach allows considering more general objectives for the OTR than maximization of an expected outcome (e.g., survival probability) by taking into account, for example, unnecessary treatment burden. As a motivating example we consider the case of oropharynx cancer treatment where unnecessary burden due to chemotherapy is to be avoided while maximizing survival chances. Assuming ignorable treatment assignment we describe Bayesian inference about the OTR including a sensitivity analysis on the unobserved partial association of the potential outcomes. We evaluate the methodology by simulations that apply Bayesian parametric and more flexible non-parametric outcome models. The proposed OTR for oropharynx cancer reduces the frequency of the more burdensome chemotherapy assignment by approximately 75\% without reducing the average survival probability. This regime thus offers a strong increase in expected quality of life of patients.

\end{abstract}

\section{Introduction}
An important question in current research in personalised medicine is how to select the best treatment based on all information available about the individual. This development is facilitated as researchers have access to increasingly large amounts of observational data from diverse sources, such as electronic patient files or extracts from CT or MRI scans and genomics. Traditional analyses from evidence-based medicine often only use sub-group analysis of treatment effects for predefined clinical baseline characteristics.  Our motivating example is primary treatment of oropharyngeal squamous cell carcinoma (OPSCC) that includes radiotherapy (RT) or RT with concomitant chemotherapy (CRT). Meta-analyses have shown an average 6.5 \%-points gain in 5-year survival for patients given CRT as opposed to RT treatment, where survival gain under CRT is higher in patients with high-stage tumours, better performance status, and the younger \parencite{pignon_2009}. These results suggest the potential of simple individualised treatment rules (e.g., give CRT in high-stage patients), but such statistics are only available for a few characteristics and it is unknown how the available characteristics interact. Another important consideration in practice is treatment burden that patients encounter due to side-effects and toxicity (e.g., CRT is seriously toxic). Toxic treatments should only be administered if a clear benefit is expected (e.g., survival instead of death). A relevant question therefore is whether unknown individualised OPSCC treatment rules exist that increase patient benefit beyond current standard of care in terms of survival and burden.

Finding a set of treatment assignment rules that optimize benefit for all patients in a population means estimating an 'optimal treatment regime' \parencite[OTR;][]{robins_optimal_2004, murphy_optimal_2003}. There is increased effort to develop OTR estimation methodology, both for dynamic (multiple decisions taken in sequence) and static (decisions taken at one point in time) regimes \parencite{zhang_robust_2012, zhang_robust_2013, zhang_estimating_2012, zajonc_bayesian_2012, zhang_robust_2013, zhao_new_2015, xu_bayesian_2016, murray_bayesian_2017, thall_bayesian_2007}. This literature assumes that data from a randomized trial are available or that data are observational and treatment assignment is ignorable \parencite{rubin_bayesian_1978}. Regardless of type of study design, the risk of misspecifying the outcome model is a substantial concern because it leads to biased estimates of treatment effects and incorrect treatment decisions. A viable approach is, therefore, to deploy statistical and machine learning techniques to approximate the unknown functional relationship between potential outcomes and patient characteristics \parencite{qian_performance_2011, taylor_reader_2015, imai_estimating_2013, xu_bayesian_2016, murray_bayesian_2017, athey_recursive_2016}. An alternative is to use treatment assignment propensities in weighted estimating equations or treatment assignment classifiers \parencite{robins_estimation_2008, zhang_estimating_2012, zhao_estimating_2012, zhao_new_2015}. These procedures require correct specification of the treatment assignment model (or known propensities), if the outcome model is misspecified.

In this paper, we consider Bayesian estimation of static OTR for dichotomous potential outcomes with observational data and develop an OTR for OPSCC treatment. Our paper extends available OTR methodology in two respects. First, we implement a Bayesian decision theoretical framework using a general loss function defined on the bivariate distribution of the potential outcomes  \parencite[for Bayesian decision theory see][]{parmigiani_decision_2009}. We select the treatment which minimizes the mean posterior loss (the Bayes decision). This is important, because loss may be related to a patient covariate even in the absence of a treatment-covariate interaction effect in, for instance, a logistic regression. We introduce the so-called conditional parametrization of the loss function which allows penalizing three clinically relevant treatment errors:  (a) a more burdensome treatment is selected although the less burdensome one leads to the same outcome (patient has unnecessary burden); (b) a less burdensome treatment is selected but the patient only has positive outcome (e.g., survival) under the more burdensome one (patient death is avoidable); (c) a more burdensome treatment is selected but the patient only survives under the less burdensome one (avoidable death and unnecessary burden). Using this approach with loss penalties for the unnecessary burden of CRT, we estimate an OTR for OPSCC that strongly reduces CRT assignment by approximately 75\% while holding average three-year survival probabilities constant at the level of the observed regime. Our approach is related to work by \textcite{li_doseschedule_2008} and \textcite{lee_bayesian_2015} who defined utilities on a joint distribution of efficacy and toxicity in the design of dose-response trials. Although similar in spirit, our focus is estimating OTR using a loss function that fully depends on the distribution of potential (efficacy) outcomes.

Second, we  quantify uncertainty about the correctness of treatment decisions and the expected loss under the OTR. This objective is hard for much OTR methodology, but facilitated in the Bayesian setting. Specifically, we are concerned with three sources of uncertainty in observational data, i.e. (i) sampling error, (ii) missing potential outcomes  \parencite{rubin_causal_2005}, and (iii) an unidentified partial correlation between the potential outcomes \parencite{ding_potential_2016}. Bayesian estimation facilitates quantifying uncertainty, because posterior distributions widen in sparse regions of the data (due to i and ii) and unidentified parameters (iii) can be handled in estimation. The potential for Bayesian inference about OTR has been recently recognized by a small body of literature \parencite{murray_bayesian_2017, zajonc_bayesian_2012, xu_bayesian_2016, arjas_optimal_2010, thall_bayesian_2007}. To avoid outcome model misspecification, we follow \textcite{murray_bayesian_2017} and \textcite{xu_bayesian_2016} who use flexible non-parametric Bayesian learning techniques for function approximation, in particular Bayesian additive regression trees \parencite[BART,][]{chipman_bart:_2010, hill_bayesian_2011}. We develop a posterior estimate of certainty on the optimal decision similar to \textcite{murray_bayesian_2017} and previously unconsidered credible intervals of the expected loss and outcome under the OTR. We take into account the non-identified partial correlation (iii) by a Bayesian sensitivity analysis. We caution that the selection of the optimal treatment can distort the nominal frequentist coverage property of credible intervals, because the optimal decision selects the posterior distribution of expected loss with the smaller mean. Bounds on the type I error rate are derived.

The remainder of this paper is structured as follows. In section \ref{sec:theory} we present the Bayesian decision theoretical framework for OTR with discrete outcomes. In section \ref{sec:estimation} we discuss Bayesian estimation and inference on the OTR using data from observational designs. In section \ref{sec:simulation} we present results on a simulation study assessing performance of the estimation approach. Section \ref{sec:application} presents the application to treatment of OPSCC.

\section{Theoretical framework} \label{sec:theory}
We consider observational data generated according to the Rubin causal model \parencite{rubin_bayesian_1978, rubin_estimating_1974}. For two treatments, random variable $W \in \{0,1\}$ denotes assignment to treatment $0$ or $1$. The observable potential outcomes are binary coded $Y=(Y(0),Y(1))$, with $Y(1)$ the outcome under $W=1$ and $Y(0)$ the outcome under $W=0$. Depending on $W$ either $Y(0)$ or $Y(1)$ is observed, where the observed outcome is $Y(W) = WY(1)+(1-W)Y(0)$. Covariates $\X \in \mathbb{R}^{p}$ are observed pre-treatment measurements of confounders and predictors of $(Y(0),Y(1))$. We assume for any observation that the potential outcomes are jointly multinomial distributed,
\begin{align}
\label{eq:bivariate_dist}
(Y(0),Y(1) | \X) \sim \text{multinom}(\theta_{00}, \theta_{10}, \theta_{01}, \theta_{11}| \X, n=1),
\end{align}
\sloppy with random parameter vector $\bm{\theta} = (\theta_{00}, \theta_{10}, \theta_{01}, \theta_{11})$. The marginal distributions are Bernoulli with $(Y(0)|\X) \sim \text{Bern}(\theta_{1+}|\X)$ and $(Y(1)|\X)\sim \text{Bern}(\theta_{+1}|\X)$ with $\theta_{1+}:= \theta_{10}+\theta_{11}$ and $\theta_{+1}:= \theta_{01}+\theta_{11}$ where $P(Y(0)=1|\X,\bm{\theta})= \E(Y(0)|\X,\bm{\theta})=\E(Y(0)|\X,\theta_{1+})$ and $P(Y(1)=1|\X,\bm{\theta})= \E(Y(1)|\X,\bm{\theta})=\E(Y(1)|\X,\theta_{+1})$.

In line with the Rubin causal model, we assume, first, ignorable treatment assignment
\parencite{rubin_bayesian_1978},
\begin{align}
\label{eq:ignorable_treatment}
Y(0), Y(1), \bm{\theta} \independent W | \X
\end{align}
meaning that the treatment assignment is conditionally independent of the potential outcomes given patient covariates. In randomized trials $Y(0), Y(1), \bm{\theta} \independent W$ holds unconditional $\X$ due to randomization of $W$. Second, we assume positivity (overlap assumption) $0 < P(W | \X) < 1$, which precludes deterministic treatment assignment. Third, we assume no interference between the treatment assignment of a patient and the outcomes of others \parencite[Stable Unit Treatment Value Assumption (SUTVA),][]{imbens_causal_2015}. Due to SUTVA all observations are identically and independently distributed according to the joint distribution of $(Y(0),Y(1),W,\X)$. The complete data are $\D=\{Y(0),Y(1),W,\X\}$; the observed data are $\D_{obs}=\{Y(W),W,\X\}$. 

\subsection{General loss function and Bayes decision} \label{sec:loss_decision}
We introduce treatment decision $a \in \{0,1\}$ for any future patient with outcomes ($Y(0),Y(1)$), where $a=0$ denotes assigning treatment 0 and $a=1$ assigning treatment 1. When $a$ is assigned, the observed outcome is $ Y(a) = (1-a) Y(0) + a Y(1)$. Without loss of generality we assume that $Y(a)=1$ denotes treatment success (e.g., survival) and $Y(a)=0$ denotes treatment failure (e.g., death). The loss associated with the treatment decision is described by a loss function:
\begin{align}
\label{eq:loss}
\Loss(Y(0),Y(1), a) := \sum_{(j,k) \in \{0,1\}^2} L_{jk}^{(a)} \mathbbm{1}_{ \{Y(0)=j,Y(1)=k\}} 
\end{align}
with $L_{jk}^{(a)}  \ge0 $ fixed loss coefficients. Our objective is to take a Bayes optimal treatment decision for each new patient with profile $\X=\x$ given a specific choice for $L^{(a)}_{jk}$. We define the expected loss for a particular value of $(\bm{\theta}, \ \x)$ and decision $a$ as
\begin{align}
\label{eq:exp_loss}
\mu_{\Loss}(a,\bm{\theta},\x) &:= \E_{Y|\x, \bm{\theta}} \big[\Loss(Y(0),Y(1), a)\big] \nonumber \\
&= \sum_{(j,k) \in \{0,1\}^2} \Loss(Y(1)=j,Y(0)=k, a) P(Y(1)=j,Y(0)=k| \x,\bm{\theta} )
\end{align}
Parameter $\mu_{\Loss}(a,\bm{\theta}, \x)$ is the estimand on which optimal decisions are based and its uncertainty is reflected by the posterior distribution across $\bm{\theta}$ denoted $\pi( \bm{\theta} | \x, \D)$ where we condition on $\X=\x$. We assume that sampling from $\pi(\bm{\theta} | \x, \D)$ is possible and defer details to section \ref{sec:estimation}. The Bayes decision $a_1^*(\x)$ minimizes the posterior expected loss for $\X=\x$,
\begin{align}
\label{eq:post_exp_loss}
a_1^*(\x) := \argmin_a \E_{\bm{\theta} | \x, \D } \big[  \mu_{\Loss}(a,\bm{\theta},\x) \big]= \argmin_a \int_{\Theta} \mu_{\Loss}(a,\bm{\theta},\x) \pi(\bm{\theta} | \x, \D) d\bm{\theta}
\end{align}
with $\Theta$ the support of $\bm{\theta}$. We now define the loss contrast
\begin{align} \label{eq:Delta_Loss}
\Delta_{\Loss}(\bm{\theta},\x) := \mu_{\Loss}(1,\bm{\theta},\x) -  \mu_{\Loss}(0,\bm{\theta},\x),
\end{align}
so that
\begin{align}
\label{eq:OTR_rule1}
a_1^*(\x) :=
\begin{cases}
1 \quad \text{if} \ \E_{ \bm{\theta} | \x, \D }  \big[ \Delta_{\Loss}(\bm{\theta},\x) | \x, \D \big] <0   \\
0 \quad \text{else}.
\end{cases}
\end{align}
Parameter $\Delta_{\Loss}(\bm{\theta},\x)$ can be viewed as the conditional treatment effect on the expected loss and its posterior density $\pi_{\bm{\theta}|\x,\D}( \Delta_{\Loss}(\bm{\theta},\x) | \x, \D)$ is a function of $\pi(\bm{\theta} | \x, \D)$. The posterior probability $\rho(\x)$ that $\mu_{\Loss}(1,\bm{\theta},\x) \le \mu_{\Loss}(0,\bm{\theta},\x)$ is 
\begin{align}
\label{eq:pcp}
\rho(\x) := P( \Delta_{\Loss}(\bm{\theta},\x) \le 0 | \x,\D  ). 
\end{align}
Probability $\rho(\x)$ quantifies the posterior certainty about the treatment decision $a_1^*(\x)=1$; $1-\rho(\x)$ may be used in the alternative case. 
It may be noted that $\rho(\x) $ can be used to form an alternative decision rule
\begin{align}
\label{eq:OTR_rule2}
a_2^*(\x) :=
\begin{cases}
1 \quad \text{if} \ \rho(\x)>0.5  \\
0 \quad \text{else}.
\end{cases}
\end{align}
This alternative uses the posterior median of $\pi_{\bm{\theta}|\x,\D}( \Delta_{\Loss}(\bm{\theta},\x) | \x,\D)$ instead of posterior mean (\ref{eq:OTR_rule1}) and is Bayes for the $0-1$ loss function
\begin{align}
\label{eq:loss2}
\Loss_2(a_2, \bm{\theta},\x) = 
\begin{cases}
1 \quad \text{if}  \ (\Delta_{\Loss}(\bm{\theta},\x) > 0 \cap a_2=1) \cup (\Delta_{\Loss}(\bm{\theta},\x) \le 0 \cap a_2=0) \\
0 \quad \text{else}.
\end{cases}
\end{align}
In the following, we write general $a^*$ and only index the parameter where needed. An OTR is defined by a rule to compute $a^*(\x)$ and $\rho(\x)$ for any $\x$ sampled from $p(\X)$. The expected loss (\ref{eq:exp_loss}) under the OTR is then given by $\mu_{\Loss}(a^*(\x),\bm{\theta}, \x )$. A second relevant estimand is the expected outcome under the OTR (e.g., survival probability). The expected outcome is
\begin{align}
\label{eq:mu_Y}
\mu_{Y}(a,\bm{\theta},\x) := \E_{Y|\x,\bm{\theta}}( Y( a )  | \x, \bm{\theta})  = (1-a)   \E(Y(0) | \x, \theta_{1+} ) + a \E(Y(1) | \x, \theta_{+1}  )   ,
\end{align}
with $ \mu_{Y}(a^*(\x),\bm{\theta},\x) $ under the OTR. This approach clearly distinguishes the expected loss and the expected outcome, which are distinct target quantities. 

\subsection{Loss functions: two special cases} \label{sec:loss_constr}
We  discuss two important special cases of the fully parametrized loss function (\ref{eq:loss}) (Table \ref{tab:loss}). We now make the additional assumption that treatment decision $a=1$ means administering the treatment that is more burdensome, for instance, because of higher toxicity or more severe side-effects (e.g., CRT). The 'marginal parametrization' penalises $a=1$ by $L_t$ and it also penalises a negative outcome (e.g., dying) under any treatment by $L_d$. 
This approach is akin to making treatment decisions based on a quality-of-life adjusted outcome measure. The associated contrast (\ref{eq:Delta_Loss}) simplifies to
\begin{align}
\label{eq:marginal_par}
\Delta_{\Loss}(\bm{\theta},\x) := L_d [\E(Y(0)|\x,\theta_{1+} ) - \E(Y(1)|\x,\theta_{+1} )] + L_t,
\end{align}
 which now only depends on the marginal probabilities $\E(Y(0)|\x,\theta_{1+} )$ and  $\E(Y(1)|\x,\theta_{+1} )$.

\vspace{0.5cm}
\begin{table}[h]
	\centering
	\begin{threeparttable}
		\caption{ Loss functions using the full, conditional, and marginal parametrizations. }
		\begin{tabular}{l c c c c c c c c} 
			
			\label{tab:loss}
			& \multicolumn{2}{c}{Full} && \multicolumn{2}{c}{Conditional} && \multicolumn{2}{c}{Marginal} \\
			\cline{2-3}
			\cline{5-6}
			\cline{8-9}
			Potential outcomes & $a=0$ & $a=1$ &&  $a=0$ & $a=1$ &&  $a=0$ & $a=1$ \\
			\hline
			$Y(0) = 0, Y(1)=0$  & $L^{(0)}_{00}$ & $L^{(1)}_{00}$ && $0$ & $L^{(1)}_{00}$  && $L_d $ & $L_d+ L_t $ \\
			$Y(0) = 0, Y(1)=1$  & $L^{(0)}_{01}$ & $L^{(1)}_{01}$ && $L^{(0)}_{01}$ & $0$  && $L_d$ & $L_{t}$ \\
			$Y(0) = 1, Y(1)=0$  & $L^{(0)}_{10}$ & $L^{(1)}_{10}$ && $0$ & $L^{(1)}_{10}$  && $0$ & $L_d+L_t$ \\
			$Y(0) = 1, Y(1)=1$  & $L^{(0)}_{11}$ & $L^{(1)}_{11}$ && $0$ & $L^{(1)}_{11}$  && $0$ & $L_t$ \\
			\hline 
		\end{tabular}
	\end{threeparttable}
\end{table} 
\vspace{0.5cm}

The 'conditional parametrization' is a flexible generalization of the marginal parametrization. Losses are defined conditional on the potential outcome strata. In each stratum, the wrong treatment decision is penalised whereas losses of the correct decisions are fixed at zero. Two types of wrong decisions emerge. First, when the potential outcomes are equal, $(Y(0) = 0, Y(1)=0)$ or $(Y(0) = 1, Y(1)=1)$, one should avoid giving the more burdensome treatment ($a=1$). Clearly, it is a central concern in practice to avoid unnecessary burden for patients. The associated losses are $L_{00}^{(1)}$ and $L_{11}^{(1)}$. Second, when the potential outcomes are not equal, $(Y(0) = 0, Y(1)=1)$ or $(Y(0) = 1, Y(1)=0)$, taking the treatment decision leading to a negative outcome (e.g., death) is wrong  ($a=0$ or $a=1$, respectively). The associated parameters $L_{01}^{(0)}$ and $L_{10}^{(1)}$ encode the loss due to fatal treatment errors, because the alternative treatment would have saved the patient. The parameters need not be equal. In particular, a patient in stratum $(Y(0) = 1, Y(1)=0)$ with $a=1$ not only receives the wrong treatment but also encounters higher treatment burden than a patient in stratum $(Y(0) = 0, Y(1)=1)$ with $a=0$ (thus $L_{01}^{(0)} \le L_{10}^{(1)}$). To the contrary, a patient in $(Y(0) = 0, Y(1)=1)$ correctly receives $a=1$ (and has burden) but now the burden is unavoidable, because the alternative would be a negative outcome (e.g., death); hence $L_{01}^{(1)}=0$ here. In comparison, the marginal parametrization assigns loss $L_t$ irrespective of potential outcomes to decision $a=1$, which includes $(Y(0) = 0, Y(1)=1)$; thus $L_{01}^{(1)}=L_t$ here. This loss is usually clinically irrelevant and should not be attributed. 

\subsection{Regime maximizing the expected outcome}
Let the decision $a^*_{\text{OTRmax}}(\x)$ denote the regime maximizing the posterior expected outcome (e.g., survival probability) with
\begin{align}
a^*_{\text{OTRmax}}(\x) &:= \argmax_a \int_{\Theta} \mu_{Y}(a,\bm{\theta},\x) \pi(\bm{\theta} | \x, \D) d\bm{\theta} \\
&= \begin{cases} \label{eq:OTR_max}
1 \quad \text{if} \ \E_{ \bm{\theta} | \x, \D }  \big[ \Delta_{Y}(\bm{\theta},\x)| \x, \D \big] > 0 \\
0 \quad \text{else},
\end{cases}
\end{align}
with 
\begin{align} \label{eq:Delta_Y}
\Delta_{Y}(\bm{\theta},\x) := \mu_{Y}(1,\bm{\theta},\x) - \mu_{Y}(0,\bm{\theta},\x).
\end{align} 
This is the Bayesian equivalent to the frequentist approach for estimating OTR using an outcome model  \parencite[e.g.][]{zhang_estimating_2012}. Regime (\ref{eq:OTR_max}) emerges as special case of the conditional and marginal loss parametrization (Table \ref{tab:loss}) when minimization of loss in (\ref{eq:post_exp_loss}) equals maximization of the expected outcome in (\ref{eq:OTR_max}). It can be shown that this situation emerges for the conditional parametrization when $L_{01}^{(0)}=L_{10}^{(1)}$ and $L_{00}^{(1)}=L_{11}^{(1)}=0$. For the marginal parametrization we have $L_t=0$. 
In both cases, being indifferent about the more burdensome treatment leads to maximizing the expected outcome; a plausible result that implies a trade-off between minimizing expected loss (patient burden) and maximizing expected outcome (survival probability). 

\section{Bayesian estimation}   \label{sec:estimation}
All estimands laid out in section \ref{sec:loss_decision} are functions of $\bm{\theta}$ and $\X$; hence our objective is to obtain $\pi(\bm{\theta}|\x,\D)$. In practice, the bivariate distribution of $(Y(0),Y(1) | \X)$ in (\ref{eq:bivariate_dist}) is not identifiable by the observed data $\D_{obs}=\{Y(W),W,\X\}$. To address this problem, we first re-parametrize (\ref{eq:bivariate_dist}) as
\begin{align}
\label{eq:bivariate_dist_repar}
(Y(0),Y(1) | \X) \sim \text{multinom}(\theta_{1+}, \theta_{+1}, \phi| \X, n=1),
\end{align}
where $\phi := \theta_{11}\theta_{00}/(\theta_{10}\theta_{01})$ is the partial odds ratio, i.e. the association of $(Y(0),Y(1))$ conditional on $\X$. If $\phi=1$, $\theta_{11} =\theta_{1+}\theta_{+1}$ so that $\theta_{10}=\theta_{1+}-\theta_{11}$, $\theta_{01}=\theta_{+1}-\theta_{11}$, $\theta_{00}=1-\theta_{11}-\theta_{10}-\theta_{01}$.  If $\phi \ne 1$, $\theta_{11} = -c^2+[(c^2-4bd)(2b)^{-1}]^{1/2}$ with $b=1-\phi,\ c=1-b(\theta_{+1}+\theta_{1+}),\ d=-\phi\theta_{+1}\theta_{1+}$ \parencite{li_doseschedule_2008}. The posterior distribution of the parameters in (\ref{eq:bivariate_dist_repar}) then factors as \parencite{ding_causal_2018}
\begin{align}
\pi(\theta_{1+},\theta_{+1}, \phi| Y(0), Y(1), W, \X) &\propto \pi(\theta_{1+},\theta_{+1}, \phi|Y(W), \X) \label{eq:post_ignorability} \\
&\propto \pi(\theta_{1+}|Y(0),W=0,\X) \pi(\theta_{+1}|Y(1),W=1,\X)p(\phi),  \label{eq:post_prior_indep} 
\end{align}
where (\ref{eq:post_ignorability}) uses ignorable treatment assignment (\ref{eq:ignorable_treatment}) and  (\ref{eq:post_prior_indep}) assumes exchangeability and prior independence. We see that the posterior $\pi(\theta_{1+},\theta_{+1}, \phi| \D)$  only depends on the observed data through the posteriors $\pi(\theta_{1+}|Y(0),W=0,\X)$ and $\pi(\theta_{+1}|Y(1),W=1,\X)$. A prior assumption on the value or the distribution of $\phi$ then identifies $\pi(\theta_{1+},\theta_{+1}, \phi| \D)$ and it enables a sensitivity analysis as explained below. Note that the posterior of $\phi$ remains equal to its prior, because it is not informed by the data. To arrive at $\pi(\bm{\theta}|\x,\D)$, Bayesian models are applied by specifying the functional relationships of $\theta_{1+}$ and $\theta_{+1}$ with $\X$. However, often no prior knowledge exists on the correct functional form of models and prior distributions cannot be chosen informatively. We then suggest applying non-parametric Bayesian function approximation using BART \parencite{chipman_bart:_2010}; see section \ref{sec:simulation}. 
Posterior inference about $\mu_{\Loss}(a^*(\x),\bm{\theta},\x)$ and $\mu_{Y}(a^*(\x),\bm{\theta},\x)$, the pointwise expected loss and outcome under the OTR, is possible. The posterior of $\mu_{\Loss}$ is given by
\begin{align}
\label{eq:posterior_mu_L}
\pi_{\bm{\theta}|\x,\D}(\mu_{\Loss}(a^*(\x),\bm{\theta},\x) | \x,\D) = \begin{cases}
\pi_{\bm{\theta}|\x,\D}(\mu_{\Loss}(1,\bm{\theta},\x) | \x,\D) \quad \text{if} \ a^*(\x)=1 \\
\pi_{\bm{\theta}|\x,\D}(\mu_{\Loss}(0,\bm{\theta},\x) | \x,\D) \quad \text{if} \ a^*(\x)=0.
\end{cases}
\end{align}
Given $a^*(\x)$, the posterior of $\mu_{Y}$, see  (\ref{eq:mu_Y}), only depends on the marginal parameters
\begin{align}
\label{eq:posterior_mu_Y}
\pi_{\bm{\theta}|\x,\D}(\mu_{Y}(a^*(\x),\bm{\theta},\x) | \x,\D) = \begin{cases}
\pi_{\theta_{1+}|\x,\D}(\mu_{Y}(1,\theta_{+1},\x) | \x,\D) \quad \text{if} \ a^*(\x)=1 \\
\pi_{\theta_{+1}|\x,\D}(\mu_{Y}(0,\theta_{1+},\x) | \x,\D) \quad \text{if} \ a^*(\x)=0.
\end{cases}
\end{align}
Point estimates are given by the posterior means and  posterior  $(1-\gamma)$-credible intervals can be derived. 

\subsection{Frequentist coverage of credible intervals}

The optimal treatment decision rule (\ref{eq:OTR_rule1}) selects $a^*$ for which the posterior mean is smallest. Therefore, the posterior mean estimates and credible intervals may be biased by optimism. We consider effects on frequentist coverage that optimism may have. Define the following $(1-\gamma)$-credible interval for $\mu_{\Loss}(1, \bm{\theta},\x )$:
\[ I(1) = [ Q(\gamma/2), Q(1-\gamma/2) ], \]
where $Q$ is the quantile function of the posterior distribution of $\mu_{\Loss}(1, \bm{\theta},\x )$. Similarly, we define a $(1-\gamma)$-credible interval $I(0)$ for $\mu_{\Loss}(0, \bm{\theta},\x )$. Let $I(a^*)$ be the interval corresponding to the optimal treatment, i.e. we select either $I(1)$ or $I(0)$ based on the decision $a^*$. 

To consider frequentist coverage, let $\bm{\theta}_0$ be the true value of the parameter, i.e. the point at which the posterior distribution of $\bm{\theta}$ conditional on $\X=\x$ contracts as the sample size tends to infinity. Suppose that the coverage of $I(1)$ for $\mu_{\Loss}(1, \bm{\theta}_0,\x )$ is $1-\alpha$:
\begin{equation}\label{coverageI1}
P( \mu_{\Loss}(1, \bm{\theta}_0,\x ) < Q(\gamma/2)) = \alpha/2, \qquad 
P( \mu_{\Loss}(1, \bm{\theta}_0,\x ) > Q(1-\gamma/2)) = \alpha/2,
\end{equation}
where $P$ refers to the distribution of the expected posterior loss $\mu_{\Loss}(1,\bm{\theta},\x)$ over repeated data sets given a specific $\x$. Similarly assume that $I(0)$ has coverage $1-\alpha$ for $\mu_{\Loss}(0, \bm{\theta}_0,\x )$. We define the coverage of $I(a^*)$ for $\mu_{\Loss}(a^*, \bm{\theta}_0,\x )$ as $P(\mu_{\Loss}(a^*, \bm{\theta}_0,\x ) \in I(a^*))$. 

The following theorem gives a crude lower bound on the coverage of $I(a^*)$. It holds in general and is similar in spirit to Bonferroni bounds.

\begin{theorem}\label{thm:bonferroni}
	Suppose the credible intervals $I(1),I(0)$ have frequentist coverage $1-\alpha$. Then the frequentist coverage of $I(a^*)$ for $\mu_{\Loss}(a^*, \bm{\theta}_0,\x)$ is at least $1-2\alpha$.
\end{theorem}

For a proof see Appendix \ref{sec:appendix_A}. Under stricter conditions we now give a more detailed bound on the coverage of $I(a^*)$. It is shown that the coverage of $I(a^*)$ can be equal to $1-\alpha$ despite the optimism.

\begin{theorem}\label{thmcredible}
	Suppose the credible intervals $I(1),I(0)$ have frequentist coverage $1-\alpha$ in the sense of \eqref{coverageI1}, suppose the distribution of the means of $I(1),I(0)$ over repeated data sets is the (correlated) bivariate normal distribution with variances $\sigma^2(1),\sigma^2(0)$, and suppose the widths of $I(1),I(0)$ are constant over repeated data sets. Then the frequentist coverage of $I(a^*)$ for $\mu_{\Loss}(a^*, \bm{\theta}_0,\x )$ is
	\begin{itemize}
		\item[(i)] equal to $1-\alpha$, if $\mu_{\Loss}(1, \bm{\theta}_0,\x )=\mu_{\Loss}(0, \bm{\theta}_0,\x )$ or $\sigma^2(1) = \sigma^2(0)$,
		\item[(ii)] strictly between $1-\alpha$ and $1-\frac{1}{2}\alpha$, if $\mu_{\Loss}(1, \bm{\theta}_0,\x )>\mu_{\Loss}(0, \bm{\theta}_0,\x )$ and $\sigma^2(1)<\sigma^2(0)$, or if $\mu_{\Loss}(1, \bm{\theta}_0,\x )<\mu_{\Loss}(0, \bm{\theta}_0,\x )$ and $\sigma^2(1)>\sigma^2(0)$,
		\item[(iii)] strictly between $1-\frac{3}{2}\alpha$ and $1-\alpha$, if $\mu_{\Loss}(1, \bm{\theta}_0,\x )>\mu_{\Loss}(0, \bm{\theta}_0,\x )$ and $\sigma^2(1) > \sigma^2(0)$, or if $\mu_{\Loss}(1, \bm{\theta}_0,\x )<\mu_{\Loss}(0, \bm{\theta}_0,\x )$ and $\sigma^2(1) < \sigma^2(0)$.
	\end{itemize}
\end{theorem}

For the proof of Theorem \ref{thmcredible} see Appendix \ref{sec:appendix_A}. In particular, case (i) indicates that the coverage is $1-\alpha$ when there is no treatment effect ($\Delta_{\Loss}$) on the expected loss.  
From the proof it follows that the bounds $1-\frac{1}{2}\alpha$ and $1-\frac{3}{2}\alpha$ in cases (ii) and (iii) are approximated only for extreme values of the means and standard deviations. Therefore, Theorem \ref{thmcredible} says that, for reasonable values of the parameters and under the assumptions of the theorem, coverage is approximately $1-\alpha$. In Section \ref{sec:simulation}, we report on a simulation assessing the empirical performance of the credible intervals under the OTR.

\subsection{Sensitivity analysis} \label{sec:sensitivity}
Inference may be sensitive to the non-identifiable parameter $\phi \in (0,\infty)$ which denotes the partial odds ratio between the potential outcomes. Prior assumptions on $\phi$ impact the joint posterior $\pi(\theta_{1+},\theta_{+1},\phi|\D)$, (\ref{eq:post_prior_indep}), and thus  $\pi(\bm{\theta}|\D)$. The parameters affected by $\phi$, however, vary depending on the loss parametrization (Table \ref{tab:sensitivity}). The marginal parametrization generally leads to non-sensitive decisions $a^*$, $a^*_{\text{OTRmax}}$ and posteriors of $\mu_{\Loss}$ and $\mu_{Y}$.
For conditional loss, decision $a^*$ (\ref{eq:OTR_rule1}) is sensitive to $\phi$, as it depends on $\pi(\bm{\theta}|\D)$.  However, using the parametrization OTRmax, decision $a^*_{\text{OTRmax}}$ only depends on the marginal probabilities, (\ref{eq:Delta_Y}), so that it is non-sensitive to $\phi$. The posteriors of $\mu_{\Loss}$ and $\mu_{Y}$ generally depend on $\phi$ via $a^*$. Even when fixing $a^*=a$, the posterior of $\mu_{\Loss}$, (\ref{eq:posterior_mu_L}), still depends on $\pi(\bm{\theta}|\D)$. The posterior of $\mu_Y$, however, is non-sensitive for fixed $a^* = a$ as it only depends on the marginal probabilities. These assertions also hold for the posteriors of $\mu_{\Loss}$ and $\mu_{Y}$ under $a^*_{\text{OTRmax}}$. For a sensitivity analysis, we use a symmetric interval around reference $\phi_0$  (e.g., $\phi_0=1$) with conservative upper and lower bounds $[u,l]$. Decisions $a^*$ are determined at $\phi=(l,\phi_0,u)$; a sensitive decision is affected by the extreme choices $[u,l]$.  For patients with sensitive decisions the optimal treatment decision cannot be determined without an assumption on $\phi$. For patients with non-sensitive decisions we first take the optimal decision and then, holding the decision fixed at $a=a^*$, we define an uniform prior on $\phi \sim U(l,u)$, so that uncertainty in $\phi$ is reflected in the posterior of $\mu_{\Loss}$.
	
\begin{table}[h]
	\centering
	\begin{threeparttable}
		\caption{ Parameter sensitivity to partial correlation $\phi$ by parametrizations. }
		\begin{tabular}{l c c c c } 
			\label{tab:sensitivity}
		    & Conditional & Marginal \\
			\hline
			\rule{0pt}{15pt}
			$a^*(\x)$ & yes & no \\
			\rule{0pt}{15pt}
			$a^*_{\text{OTRmax}}(\x)$ & no & no \\
			\rule{0pt}{15pt}
			$\mu_{\Loss}(a^*(\x)=a,\bm{\theta},\x)$ & yes & no \\
			\rule{0pt}{15pt}
			$\mu_{Y}(a^*(\x)=a,\bm{\theta},\x)$ & no & no \\
			\rule{0pt}{15pt}
			$\mu_{\Loss}(a^*_{\text{OTRmax}}(\x)=a,\bm{\theta},\x)$ & yes & no \\
			\rule{0pt}{15pt}
			$\mu_{Y}(a^*_{\text{OTRmax}}(\x)=a,\bm{\theta},\x)$ & no & no \\
			\hline 		
		\end{tabular}
	\end{threeparttable}
\end{table}

\section{Simulation study}      \label{sec:simulation}
The objective of the simulation study was to evaluate the performance of the method in two settings: (1) functional approximation of the marginal probabilities by Bayesian Additive Regression Trees \parencite[BART;][]{chipman_bart:_2010} and (2) estimation by logistic regression models. We evaluated effects on the
\begin{enumerate}
	\item Accuracy of the Bayes decisions $a^*_1(\x)$, 
	\item Bias of posterior mean estimates of expected loss and outcome under the OTR ($\mu_{\Loss}(a^*(\x),\bm{\theta},\x)$ and$\ \mu_{Y}(a^*(\x),\bm{\theta},\x)$), and
	\item Interval width and frequentist coverage probability of credible intervals.
\end{enumerate}
We considered the impact of five factors: the strength of heterogeneity in the treatment effects $\Delta_Y(\bm{\theta},X_1)$ across a covariate $X_1$, the strength of confounding in treatment assignment, sample size, the type of loss function, and the presence of unrelated noise variables in $\X$.
\subsection{Simulation set-up}
We first generated one variable $X_1 \sim U (-1,1)$ and specified the true outcome models as
\begin{align}
\E(Y(0) | X_1; \bm{\beta}_0) &= \text{expit}(\beta_{00} + \beta_{10} X_1 + \beta_{20} X_1^2 + \beta_{30} X_1^3) \nonumber\\
\E(Y(1) | X_1; \bm{\beta}_1) &= \text{expit}(\beta_{01} + \beta_{11} X_1 + \beta_{21} X_1^2 + \beta_{31} X_1^3) \label{eq:true_mod},
\end{align}
\sloppy with $\text{expit}(x)=1/(1+\exp(-x))$ and $\bm{\beta}_j = (\beta_{0j},\beta_{1j},\beta_{2j},\beta_{3j})^T$, $j=0,1$. Models (\ref{eq:true_mod}) determine the true conditional marginal probabilities $(\theta_{+1},\theta_{1+})|X_1$; hence from (\ref{eq:Delta_Y}) we have $\Delta_Y(\bm{\theta},X_1) = \E(Y(1) | X_1; \bm{\beta}_1) - \E(Y(0) | X_1; \bm{\beta}_0)$. We fixed $\bm{\beta}_1$ at $\bm{\beta}_1=(0.457,3.185,-1.593,-2.124)^T$ and set $\bm{\beta}_0$ at $\bm{\beta}_0=( \beta_{01},-\beta_{11},\beta_{21},-\beta_{31})^T$ (strong heterogeneity), $\bm{\beta}_0=(0.457,1.343,-1.430, -1.217)^T$ (mild heterogeneity), or $\bm{\beta}_0=\bm{\beta}_1$ (no heterogeneity). This choice for the parameters in (\ref{eq:true_mod}) leads to the realistic scenario of a biomarker impacting treatment success as shown in Figure \ref{fig:sim_setup}. 

\begin{figure}[t!]  
	\begin{subfigure}{0.48\textwidth}
		\includegraphics[trim= 4.5cm 1cm 5.5cm 2cm, clip=true, width=\linewidth]{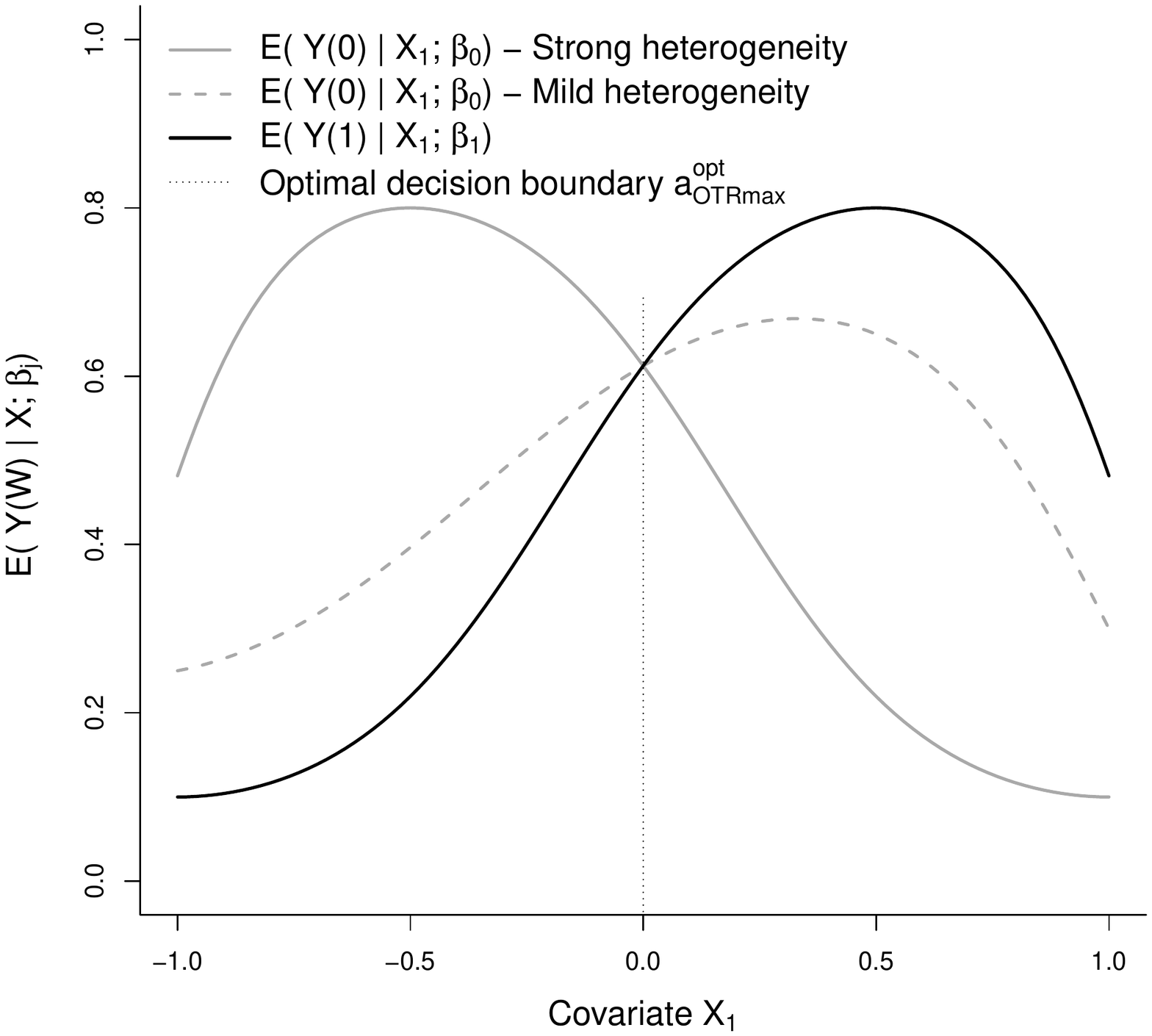}
		\caption{True functional forms $\E(Y(W=j)|X_1;\bm{\beta}_j)$} \label{fig:sim_setup}
	\end{subfigure}\hspace*{\fill}
	\begin{subfigure}{0.48\textwidth}
		\includegraphics[trim= 4.5cm 1cm 5.5cm 2cm, clip=true, width=\linewidth]{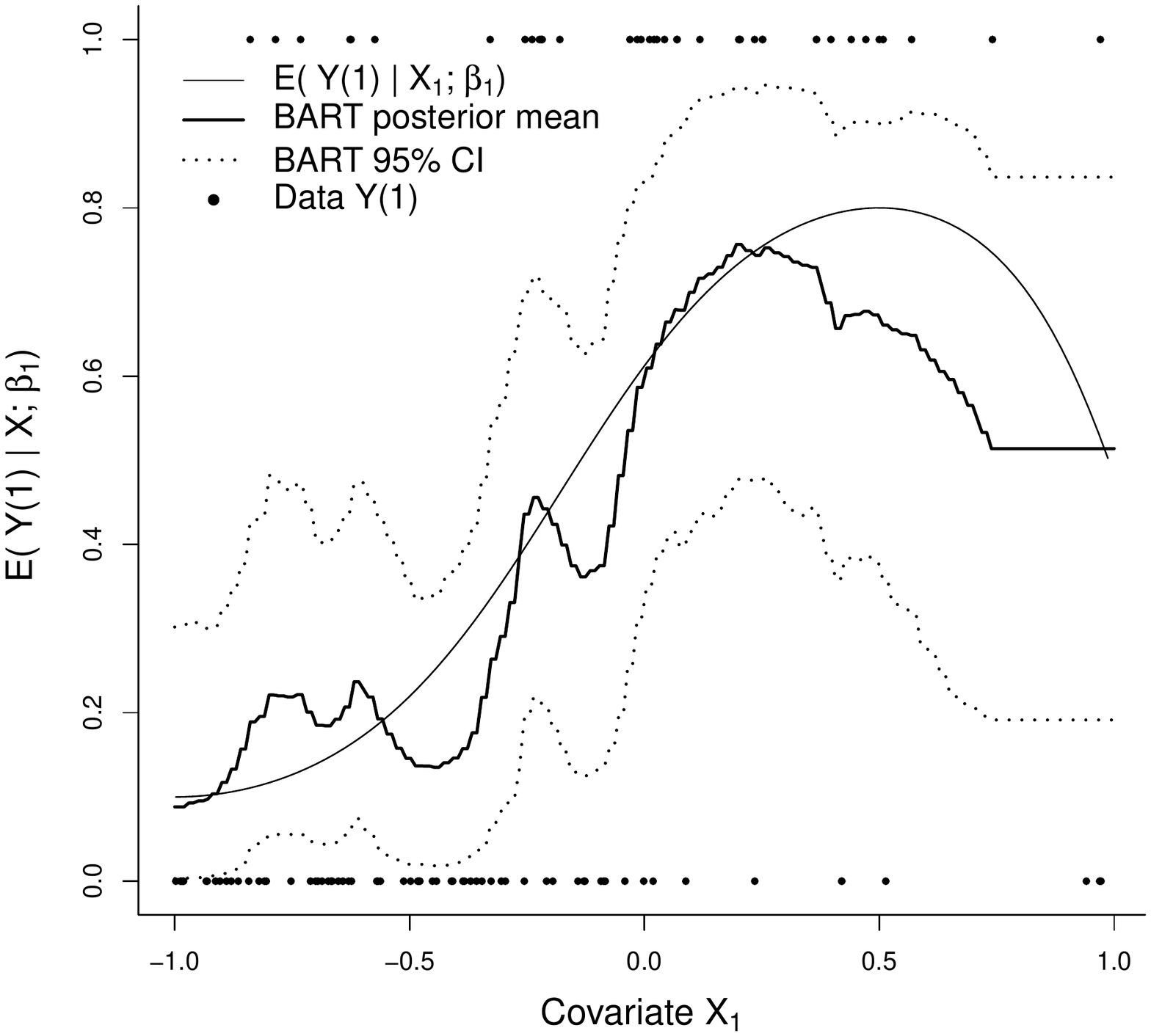}
		\caption{BART approximation to $\E(Y(1)|X_1;\bm{\beta}_1)$} \label{fig:_bart_y_approx}
	\end{subfigure}
	
	\medskip
	\begin{subfigure}{0.48\textwidth}
		\includegraphics[trim= 4.5cm 1cm 5.5cm 2cm, clip=true, width=\linewidth]{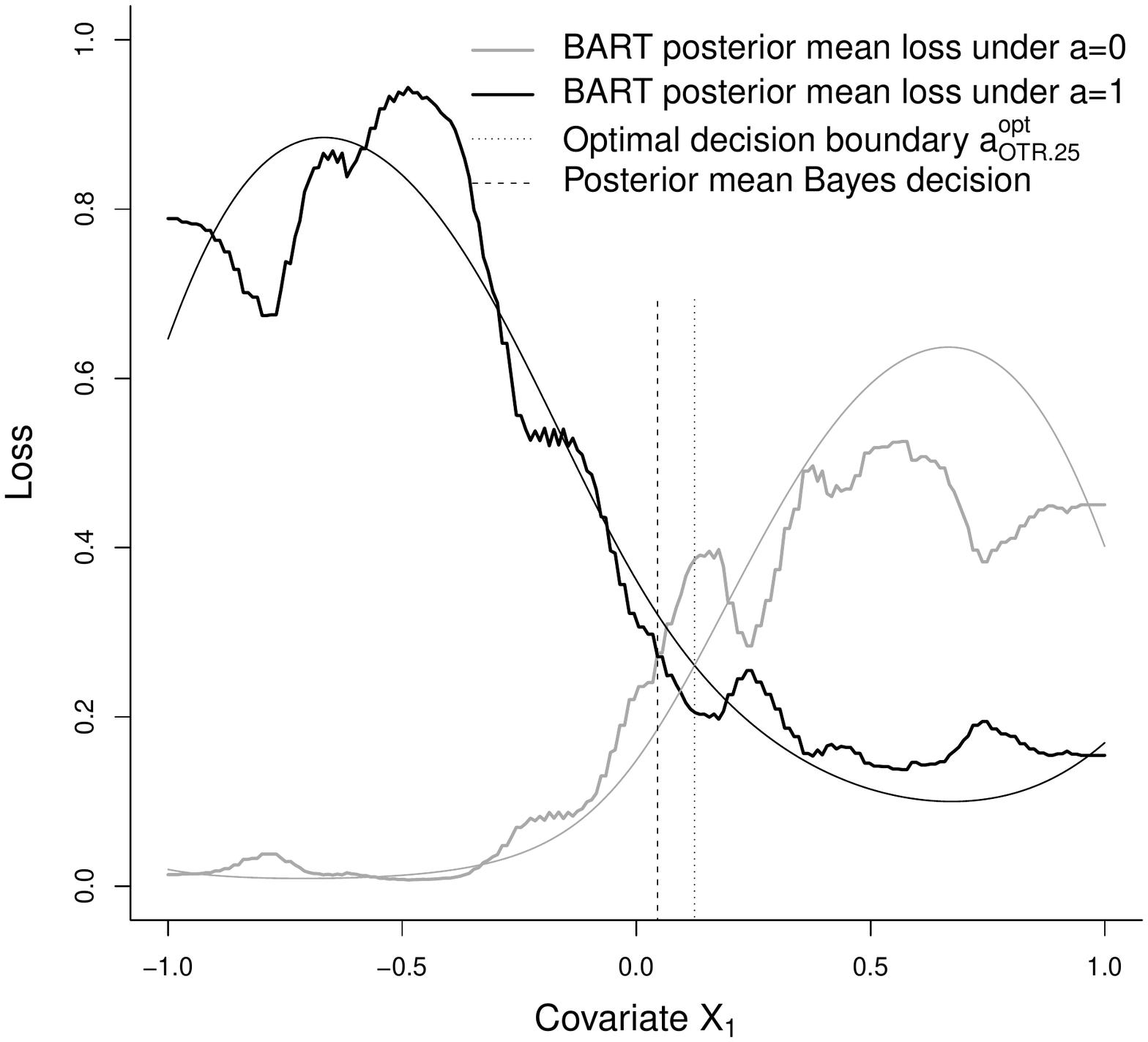}
		\caption{BART approximation to loss $\mu_{\Loss}(a,\bm{\theta}_0,X_1)$} \label{fig:_bart_L_approx}
	\end{subfigure}\hspace*{\fill}
	\begin{subfigure}{0.48\textwidth}
		\includegraphics[trim= 4.5cm 1cm 5.5cm 2cm, clip=true, width=\linewidth]{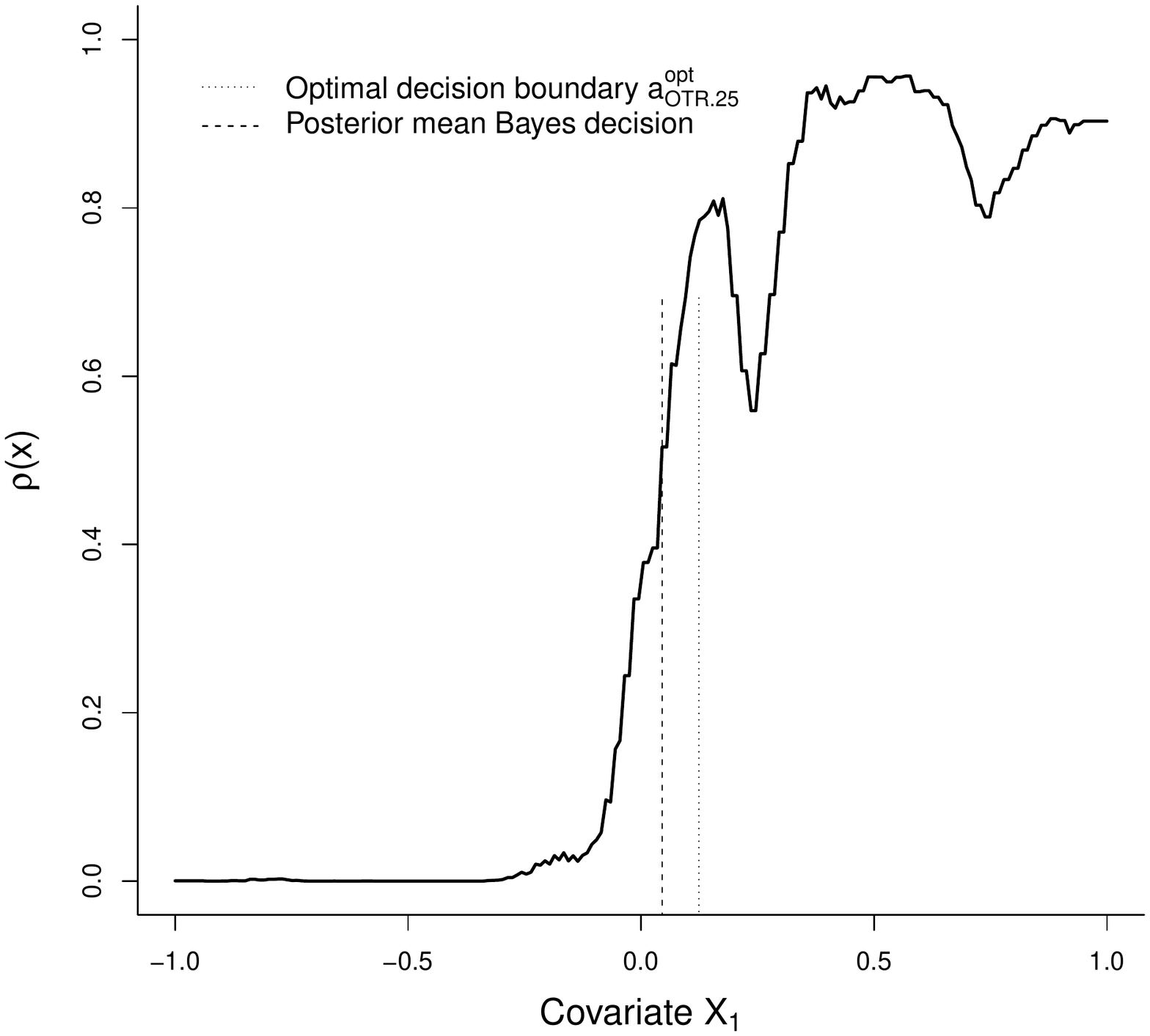}
		\caption{Posterior probability $\rho(X_1)$} \label{fig:_bart_pcp} 
	\end{subfigure}
	\caption{Illustration of simulation set-up and approximation with BART. (a) Marginal probabilities for $(Y(0), Y(1))$ by degree of treatment effect heterogeneity, (b) approximation of the marginal probability of $Y(1)$ with BART, (c) estimated decision boundary at intersection of the posterior means of losses, and (d) associated posterior probability $\rho(X_1)$. Figures (b)-(d) use $\lambda=-\log(3)$ and $n=250$ under strong treatment effect heterogeneity and OTR.25 loss. }
\end{figure}

We simulated $W$ from a Bernoulli distribution with success probability
\begin{align}
\label{eq:treatment_assignment_mech}
\E(W | X_1; \lambda) &= \text{expit} \bigg( \frac{\lambda (X_1-\bar{X_1})}{\text{sd}(X_1)} \bigg)
\end{align}
\sloppy where $\bar{X_1}$ and $\text{sd}(X_1)$ are sample mean and standard deviation. We chose $\lambda \in \{ -\log(3),0,\log(3) \}$, where $\lambda = \pm \log(3)$ caused strong selectivity at the boundaries of the support of $X_1 \in (-1,1)$. Condition $\lambda=0$ represented randomized assignment.

We considered two sets $\X$, first, $\X=X_1$ and, second, $\X= ( X_1,X_2,X_3,X_4,X_5, X_6 )$, where $X_k \sim U (-1,1)$, $k=2,...,6$, are independent noise variables unrelated to outcome and treatment assignment; let $q \in \{0,5\}$ denote the number of noise variables. Furthermore, we varied sample size $n \in \{250,500,1000\}$ and considered two combinations of loss functions and partial correlations in the conditional parametrization (Table \ref{tab:loss}), i.e. $(L^{(1)}_{00}, L^{(0)}_{01}, L^{(1)}_{10}, L^{(1)}_{11} ) = (0,1,1,0)$ with $\phi=1$ and $(0.25,1,1.25,0.25)$ with $\phi=5$. The former parametrization leads to OTRmax (\ref{eq:OTR_max}); the latter is referred to as OTR.25 in the following, where we penalised the loss due to unnecessary burden by 25\% of the loss of a wrong treatment decision.
Let $a^{opt}$ denote the optimal assignment rule determined on the marginal probabilities implied by true model (\ref{eq:true_mod}) and a loss function. Figures of the resulting true loss functions and optimal decisions  $a^{opt}$ for strong, mild, and no treatment effect heterogeneity are provided in Appendix \ref{sec:appendix_B} for both OTRmax and OTR.25. For the OTRmax regime maximizing expected outcome, (\ref{eq:OTR_max}), the optimal decision boundary is also shown in Figure \ref{fig:sim_setup}. In particular, $a^{opt}_{\text{OTRmax}}(X_1)= \mathbbm{1}_{ \{X_1>0\}}$, i.e. where $\Delta_Y(\bm{\theta},X_1)>0$ the decision should be $a=1$ for mild or strong heterogeneity and $a=0$ elsewhere. A special case emerges when $\bm{\beta}_0=\bm{\beta}_1$ (no heterogeneity), i.e. choosing any treatment has same marginal outcome probability (and loss) for any $X_1$. Then infinitely many optimal regimes $a^{opt}_{\text{OTRmax}}(X_1)$ exist. In this case, we chose $a^{opt}_{\text{OTRmax}}=0$ as a benchmark to compare Bayes decisions $a^*_1(X_1)$ against. 

\subsection{Function approximation with BART}
BART are flexible non-parametric function approximators with similarities to random forests \parencite{breiman_random_2001} and gradient boosting \parencite{friedman_greedy_2001}. However, BART offer the additional advantage of inference on the posterior expected loss, posterior expected outcome, and the posterior probability of a correct treatment decision. We illustrate decision taking with BART in Figures \ref{fig:_bart_y_approx} to \ref{fig:_bart_pcp}.  Figure \ref{fig:_bart_y_approx} shows approximation and credible intervals for $\E(Y(1)|X_1;\bm{\beta}_1)$ for a single data set in the simulation with $\lambda=-\log(3)$ and $n=250$. Intervals widen in the sparse region of $(Y(1),X_1)$ on the right, which is caused by selective treatment assignment. Figure \ref{fig:_bart_L_approx} shows true loss functions under $a=0,1$ for the OTR.25 condition and strong treatment effect heterogeneity. The optimal decision boundary $a_{\text{OTR.25}}^{opt}$ is determined on the true losses and Bayes decision  $a_{\text{OTR.25}}^{*}$ (\ref{eq:OTR_rule1}) is determined on the posterior means. For posterior probabilities $\rho(X_1)$ (\ref{eq:pcp}) see Figure \ref{fig:_bart_pcp}.  Close to the Bayes decision boundary there is higher uncertainty, reflected by $\rho$ close to 0.5. 

For dichotomous outcomes $Y$, BART uses a probit sum-of-trees model consisting of $t=1,...,T$ separate regression trees $g$ modelling latent variable $Z_i = \sum_{t=1}^T g(\x_i; \mathcal{T}_t, \mathcal{M}_t) + \epsilon_i$, where $\epsilon_i \sim \mathcal{N}(0, 1)$ and $Y_i= \mathbbm{1}_{ \{ Z_i>0 \} }$. Function $g$ returns the leaf mean associated with observation $\x_i$, where $\mathcal{T}_t$ denotes a tree's decision rules, and $\mathcal{M}_t=\{\nu_{1t},...,\nu_{mt},...,\nu_{Mt} \}$ the vector of $M$ leaf node means. BART specifies prior distributions $p(\mathcal{T}_t)$ and $p(\nu_{mt} |\mathcal{T}_t )$ regularizing the depth of trees and the extremity of leave means to prevent BART from overfitting. Prior $p(\mathcal{T}_t)$ assigns probability $1-\kappa(1+\delta)^{-\eta}$ to the event that the node at depth $\delta$ is the end node with defaults ($\kappa=0.95$, $\eta=2$) strongly favouring short trees with mode at $\delta=2$. Uniform priors are used for splitting variable assignment and splitting values. Prior $p(\nu_{mt} |\mathcal{T}_t )$ is $N(0,\sigma^2)$ with $\sigma^2=3k/\sqrt{T}$ effectively shrinking the leave means to zero to limit the influence of the individual nodes on the overall fit (recommended $k=2$). The posterior distribution $\pi( (\mathcal{T}_1, \mathcal{M}_1),...,(\mathcal{T}_T, \mathcal{M}_T) | Y )$ is simulated using an MCMC backfitting algorithm with data augmentation for $Z_i=Z$ \parencite{chipman_bart:_2010}. Our implementation follows the procedure described in section \ref{sec:estimation}. We modelled observed $Y(0)$ and $Y(1)$ separately leading to two posterior distributions $\pi( (\mathcal{T}_1, \mathcal{M}_1),...,(\mathcal{T}_T, \mathcal{M}_T) | Y(0),X )$ and $\pi( (\mathcal{T}_1, \mathcal{M}_1),...,(\mathcal{T}_T, \mathcal{M}_T) | Y(1),X )$. For each draw from the posteriors, $Z$ is obtained and evaluated at $\Phi(Z)$, where $\Phi$ is the cumulative normal distribution, yielding samples from $\pi(\theta_{1+}|Y(0),\X)$ and $\pi(\theta_{+1}| Y(1),\X)$, respectively. 

\subsection{Implementation}
BART was implemented in \R{} using \pkg{pbart} from package \pkg{BART} with settings $k=2$, $\kappa=0.95$, $\eta=2$, and $T=50$, recommended by \textcite{chipman_bart:_2010}. We compared the performance of BART to Bayesian logistic regression. Without noise variables in the set $\X=X_1$ ($q=0$)  the models were (correctly) specified as in (\ref{eq:true_mod}). With noise variables ($q=5$), the models were specified as $\E(Y(j) | \X; \bm{\beta}_j) = \text{expit} ( \sum_{k=1}^6 \sum_{p=1}^3 \beta_{jkp} X^p_k )$, $j=0,1$. In both cases, we imposed an improper uniform prior, $p(\bm{\beta}_0,\bm{\beta}_1) \propto \text{const}$. We used \R{} function \pkg{mcmclogit} from package \pkg{MCMCpack} with default settings, which implements a Metropolis algorithm to generate samples from the posterior of ($\bm{\beta}_0,\bm{\beta}_1$). We ran Monte Carlo simulations by generating $K=10^3$ data sets for each factorial combination of treatment effect heterogeneity in $\bm{\beta}_0$ (\ref{eq:true_mod}), strength of selection mechanism in $\lambda$ (\ref{eq:treatment_assignment_mech}), sample size $n$, loss function (OTRmax or OTR.25), and number of independent noise variables ($q$) in $\X$. For each data set we generated $5 \times 10^3$ samples from the posterior distributions after $10^3$ samples of burn-in. We then determined for all observed $\x_{ik}$, where $i=1,....,n$ the  observations in Monte Carlo data set $k=1,...,K$, the Bayes decisions $a_1^*(\x_{ik})$ 
, the posterior mean estimates of $\mu_{\Loss}(a_1^*(\x_{ik}),\bm{\theta},\x_{ik})$ and $\ \mu_{Y}(a_1^*(\x_{ik}),\bm{\theta},\x_{ik})$, and the 95 \% credible intervals. Odds ratio $\phi$ was assumed known. We report the Monte Carlo estimates of the sample-average of the bias $B = (nK)^{-1} \sum_{ik} \hat{\mu}_{ik} - \mu_{ik}$, where $\hat{\mu}_{ik} $ is the posterior mean estimate of $\mu(a_1^*(\X_{ik}),\bm{\theta},\x_{ik})$ and $\mu_{ik} $ its true value determined by (\ref{eq:true_mod}); the width $ \omega = (nK)^{-1} \sum_{i,k} u_{ik} - l_{ik} $, where $u_{ik}$ and $l_{ik}$ the upper and lower limits of credible interval $i$ in data set $k$; the coverage probability of the credible intervals $C= (nK)^{-1} \sum_{i,k} \mathbbm{1}{ \{\mu_{ik} \in [l_{ik},u_{ik}] \} }$; and the accuracy of decision $a_1^*$ with $A= (nK)^{-1} \sum_{i,k}  \mathbbm{1} \{a_1^*(\x_{ik}) = a^{opt}(\x_{ik}) \}$.

\subsection{Results}

Table \ref{tab:simres_OTRmax} compares the correctly specified logistic regression models (\ref{eq:true_mod}) with uninformative priors and $q=0$  noise variables to BART with $q=5$ for OTRmax loss (with $\phi=1$). In Appendix \ref{sec:appendix_B}, we include results for the  logistic regression model with $q=5$ and BART with $q=0$. For the correctly specified logistic model ($q=0$), we found a bias in posterior mean estimates of expected loss and outcome under the OTR that increased at mild or no treatment effect heterogeneity and decreased in larger samples. The direction of the bias was positive in the sense that outcome probability was overestimated and loss was underestimated due to optimism. Credible interval coverage probability increased with sample size, where $C_Y$ reached roughly 95\% at $n=1000$ in most cases and $C_{\Loss}$ occasionally was slightly smaller. Interval width decreased with sample size, as expected; however we also found a difference across selectivity conditions $\lambda$, with $\lambda=-\log(3)$ causing greater average interval width than under $\lambda=1$ and $\lambda=\log(3)$. This finding illustrates the important role played by selectivity, which may cause sparsity and thus wider intervals. Accuracy $A$ of treatment decisions decreased with weaker treatment effect heterogeneity and increased with sample size, ranging from 0.964 ($n=250$, $\lambda=-\log(3)$) to 0.987 ($n=1000$, $\lambda=0$ and $\lambda=\log(3)$)  for strong heterogeneity and 0.823 to 0.947 for mild heterogeneity. Without heterogeneity, ($\bm{\beta}_0 = \bm{\beta}_1$ in (\ref{eq:true_mod})), accuracy is not defined for OTRmax (\ref{eq:OTR_max}) and we used $a^{opt}=0$ as an arbitrary benchmark to compare Bayes decisions against. About half of decisions were taken in favour of this benchmark (the others being $a^*=1$), where higher sample size yielded $A$ closer to 0.50 reflecting that Bayes decisions are taken randomly. Note that interval coverage was not impacted despite this result and the related optimism bias.

\begin{sidewaystable}  [p]
	\centering
	\caption{Simulation results for correctly specified Bayesian logistic model with uninformative priors (q=0 noise variables) and BART (q=5) for the OTRmax loss function with $\phi=1$ (Het.: treatment effect heterogeneity, $\lambda$: selectivity odds ratio, $n$: sample size, $B$: Average bias of OTR posterior mean estimate, $\omega$: Average width of credible intervals, $C$: Average coverage probability of 95\% credible intervals, $A$: Accuracy of assignment.)}	
	\begin{tabular}{cccccccccccccccccc}	\label{tab:simres_OTRmax}
		&&& \multicolumn{7}{c}{Logistic regression model ($q=0$)} && \multicolumn{7}{c}{BART ($q=5$)} \\
		\cline{4-10}
		\cline{12-18}
		Het. &$\lambda$ &  n & $B_{L}$ & $B_Y$ & $\omega_{L}$ & $\omega_{Y}$ & $C_{L}$ & $C_Y$ & $A$ &  & $B_{L}$ & $B_Y$ & $\omega_{L}$ & $\omega_{Y}$ & $C_{L}$ & $C_Y$ & $A$ \\ \hline
		\multirow{9}{*}{\rotatebox[origin=c]{90}{Strong}} &
		\multirow{3}{*}{$-\log(3)$} & 250 & -0.006 & 0.006 & 0.114 & 0.374 & 0.897 & 0.915 & 0.964 &   & 0.008 & -0.034 & 0.29 & 0.654 & 0.992 & 0.99 & 0.928 \\
		&&500 & -0.002 & 0.002 & 0.081 & 0.282 & 0.93 & 0.941 & 0.979 &   & 0.007 & -0.021 & 0.242 & 0.579 & 0.993 & 0.989 & 0.945 \\
		&&1000 & -0.001 & 0.001 & 0.058 & 0.206 & 0.937 & 0.941 & 0.986 &   & 0.007 & -0.017 & 0.204 & 0.508 & 0.994 & 0.991 & 0.958 \\ \cline{2-18}
		&\multirow{3}{*}{$0$} & 250 & -0.004 & 0.003 & 0.107 & 0.289 & 0.923 & 0.937 & 0.969 &   & 0.01 & -0.027 & 0.293 & 0.623 & 0.994 & 0.992 & 0.933 \\
		&&500 & -0.001 & 0.001 & 0.077 & 0.209 & 0.941 & 0.947 & 0.981 &   & 0.008 & -0.022 & 0.244 & 0.553 & 0.995 & 0.993 & 0.949 \\
		&&1000 & -0.001 & 0.001 & 0.054 & 0.149 & 0.942 & 0.949 & 0.987 &   & 0.008 & -0.019 & 0.204 & 0.481 & 0.996 & 0.994 & 0.958 \\ \cline{2-18}
		&\multirow{3}{*}{$\log(3)$} & 250 & -0.006 & 0.007 & 0.123 & 0.256 & 0.9 & 0.937 & 0.967 &   & 0.019 & -0.018 & 0.322 & 0.603 & 0.988 & 0.99 & 0.926 \\
		&&500 & -0.002 & 0.003 & 0.089 & 0.182 & 0.922 & 0.941 & 0.979 &   & 0.016 & -0.02 & 0.263 & 0.535 & 0.994 & 0.993 & 0.949 \\
		&&1000 & -0.001 & 0.001 & 0.063 & 0.129 & 0.941 & 0.946 & 0.987 &   & 0.012 & -0.019 & 0.215 & 0.464 & 0.996 & 0.995 & 0.959 \\ \hline
		\multirow{9}{*}{\rotatebox[origin=c]{90}{Mild}} &
		\multirow{3}{*}{$-\log(3)$} & 250 & -0.031 & 0.031 & 0.195 & 0.373 & 0.902 & 0.939 & 0.823 &   & -0.03 & 0.052 & 0.393 & 0.663 & 0.994 & 0.991 & 0.768 \\
		&&500 & -0.013 & 0.015 & 0.144 & 0.286 & 0.932 & 0.951 & 0.9 &   & -0.024 & 0.043 & 0.347 & 0.592 & 0.993 & 0.992 & 0.802 \\
		&&1000 & -0.006 & 0.006 & 0.103 & 0.209 & 0.942 & 0.946 & 0.94 &   & -0.016 & 0.028 & 0.305 & 0.52 & 0.995 & 0.994 & 0.837 \\ \cline{2-18}
		&\multirow{3}{*}{$0$} & 250 & -0.024 & 0.026 & 0.178 & 0.299 & 0.909 & 0.931 & 0.857 &   & -0.028 & 0.043 & 0.401 & 0.642 & 0.995 & 0.992 & 0.749 \\
		&&500 & -0.009 & 0.009 & 0.13 & 0.216 & 0.937 & 0.947 & 0.913 &   & -0.021 & 0.033 & 0.353 & 0.57 & 0.995 & 0.994 & 0.797 \\
		&&1000 & -0.004 & 0.004 & 0.093 & 0.154 & 0.941 & 0.944 & 0.947 &   & -0.013 & 0.02 & 0.306 & 0.496 & 0.996 & 0.996 & 0.84 \\ \cline{2-18}
		&\multirow{3}{*}{$\log(3)$} & 250 & -0.031 & 0.032 & 0.201 & 0.284 & 0.904 & 0.923 & 0.836 &   & -0.033 & 0.062 & 0.415 & 0.637 & 0.991 & 0.985 & 0.691 \\
		&&500 & -0.014 & 0.014 & 0.15 & 0.197 & 0.93 & 0.94 & 0.899 &   & -0.022 & 0.041 & 0.367 & 0.563 & 0.994 & 0.99 & 0.759 \\
		&&1000 & -0.005 & 0.006 & 0.109 & 0.136 & 0.939 & 0.938 & 0.941 &   & -0.012 & 0.026 & 0.318 & 0.486 & 0.995 & 0.993 & 0.813 \\ \hline
		\multirow{9}{*}{\rotatebox[origin=c]{90}{None}} &
		\multirow{3}{*}{$-\log(3)$} & 250 & -0.049 & 0.051 & 0.205 & 0.303 & 0.892 & 0.926 & 0.471 &   & -0.038 & 0.076 & 0.404 & 0.628 & 0.992 & 0.984 & 0.59 \\
		&&500 & -0.034 & 0.035 & 0.157 & 0.222 & 0.915 & 0.941 & 0.482 &   & -0.035 & 0.06 & 0.363 & 0.552 & 0.994 & 0.99 & 0.563 \\
		&&1000 & -0.024 & 0.025 & 0.116 & 0.16 & 0.936 & 0.946 & 0.496 &   & -0.031 & 0.049 & 0.321 & 0.476 & 0.995 & 0.992 & 0.532 \\ \cline{2-18}
		&\multirow{3}{*}{$0$} & 250 & -0.042 & 0.042 & 0.188 & 0.277 & 0.904 & 0.94 & 0.492 &   & -0.037 & 0.062 & 0.405 & 0.612 & 0.995 & 0.991 & 0.502 \\
		&&500 & -0.029 & 0.028 & 0.141 & 0.2 & 0.925 & 0.944 & 0.5 &   & -0.035 & 0.051 & 0.361 & 0.538 & 0.995 & 0.993 & 0.503 \\
		&&1000 & -0.021 & 0.021 & 0.103 & 0.142 & 0.934 & 0.949 & 0.499 &   & -0.031 & 0.044 & 0.318 & 0.466 & 0.996 & 0.994 & 0.497 \\ \cline{2-18}
		&\multirow{3}{*}{$\log(3)$} & 250 & -0.05 & 0.05 & 0.204 & 0.303 & 0.888 & 0.93 & 0.531 &   & -0.038 & 0.074 & 0.405 & 0.628 & 0.994 & 0.986 & 0.402 \\
		&&500 & -0.034 & 0.035 & 0.157 & 0.222 & 0.918 & 0.94 & 0.51 &   & -0.035 & 0.061 & 0.362 & 0.551 & 0.994 & 0.99 & 0.439 \\
		&&1000 & -0.024 & 0.024 & 0.116 & 0.158 & 0.928 & 0.945 & 0.504 &   & -0.031 & 0.049 & 0.32 & 0.477 & 0.995 & 0.992 & 0.458 \\ \hline
	\end{tabular}
\end{sidewaystable}

\begin{sidewaystable}  [p]
	\centering
	\caption{Simulation results for correctly specified Bayesian logistic model with uninformative priors (q=0 noise variables) and BART (q=5) for the OTR.25 loss function with $\phi=5$ (Het.: treatment effect heterogeneity, $\lambda$: selectivity odds ratio, $n$: sample size, $B$: Average bias of OTR posterior mean estimate, $\omega$: Average width of credible intervals, $C$: Average coverage probability of 95\% credible intervals, $A$: Accuracy of assignment.)}	
	\begin{tabular}{cccccccccccccccccc}	\label{tab:simres_OTR25}
		&&& \multicolumn{7}{c}{Logistic regression model ($q=0$)} && \multicolumn{7}{c}{BART ($q=5$)} \\
		\cline{4-10}
		\cline{12-18}
		Het. &$\lambda$ &  n & $B_{L}$ & $B_Y$ & $\omega_{L}$ & $\omega_{Y}$ & $C_{L}$ & $C_Y$ & $A$ &  & $B_{L}$ & $B_Y$ & $\omega_{L}$ & $\omega_{Y}$ & $C_{L}$ & $C_Y$ & $A$ \\ \hline
		\multirow{9}{*}{\rotatebox[origin=c]{90}{Strong}} &
		\multirow{3}{*}{$-\log(3)$} & 250 & -0.005 & 0.009 & 0.123 & 0.372 & 0.918 & 0.920 & 0.96 &   & 0.009 & -0.027 & 0.310 & 0.653 & 0.991 & 0.991 & 0.919 \\
		&&500 & -0.001 & 0.002 & 0.086 & 0.281 & 0.937 & 0.937 & 0.974 &   & 0.009 & -0.018 & 0.256 & 0.581 & 0.992 & 0.989 & 0.940 \\
		&&1000 & 0 & 0.001 & 0.060 & 0.205 & 0.948 & 0.947 & 0.984 &   & 0.009 & -0.015 & 0.212 & 0.511 & 0.993 & 0.991 & 0.954 \\ \cline{2-18}
		&\multirow{3}{*}{$0$} & 250 & -0.004 & 0.006 & 0.108 & 0.290 & 0.929 & 0.939 & 0.965 &   & 0.013 & -0.022 & 0.314 & 0.625 & 0.993 & 0.992 & 0.923 \\
		&&500 & -0.001 & 0.002 & 0.075 & 0.210 & 0.942 & 0.947 & 0.978 &   & 0.011 & -0.019 & 0.257 & 0.555 & 0.994 & 0.993 & 0.941 \\
		&&1000 & 0 & -0.001 & 0.052 & 0.150 & 0.947 & 0.951 & 0.985 &   & 0.010 & -0.018 & 0.210 & 0.484 & 0.995 & 0.994 & 0.954 \\ \cline{2-18}
		&\multirow{3}{*}{$\log(3)$} & 250 & -0.004 & 0.007 & 0.124 & 0.261 & 0.908 & 0.937 & 0.962 &   & 0.014 & -0.007 & 0.343 & 0.607 & 0.986 & 0.987 & 0.904 \\
		&&500 & -0.001 & 0.004 & 0.085 & 0.185 & 0.930 & 0.939 & 0.976 &   & 0.016 & -0.016 & 0.274 & 0.539 & 0.992 & 0.992 & 0.937 \\
		&&1000 & 0 & 0.001 & 0.059 & 0.131 & 0.946 & 0.944 & 0.985 &   & 0.014 & -0.017 & 0.22 & 0.467 & 0.994 & 0.994 & 0.954 \\  \hline
		\multirow{9}{*}{\rotatebox[origin=c]{90}{Mild}} &
		\multirow{3}{*}{$-\log(3)$} & 250 & -0.016 & 0.024 & 0.196 & 0.356 & 0.927 & 0.924 & 0.846 &   & -0.012 & 0.048 & 0.411 & 0.661 & 0.995 & 0.991 & 0.816 \\
		&&500 & -0.011 & 0.014 & 0.146 & 0.266 & 0.935 & 0.936 & 0.867 &   & -0.009 & 0.036 & 0.364 & 0.591 & 0.995 & 0.993 & 0.833 \\
		&&1000 & -0.007 & 0.008 & 0.106 & 0.193 & 0.951 & 0.954 & 0.892 &   & -0.007 & 0.026 & 0.318 & 0.516 & 0.996 & 0.994 & 0.848 \\ \cline{2-18}
		&\multirow{3}{*}{$0$} & 250 & -0.015 & 0.020 & 0.185 & 0.303 & 0.926 & 0.930 & 0.849 &   & -0.007 & 0.035 & 0.427 & 0.645 & 0.995 & 0.992 & 0.792 \\
		&&500 & -0.009 & 0.011 & 0.136 & 0.22 & 0.940 & 0.944 & 0.874 &   & -0.006 & 0.0260 & 0.375 & 0.575 & 0.995 & 0.994 & 0.809 \\
		&&1000 & -0.006 & 0.007 & 0.098 & 0.158 & 0.948 & 0.952 & 0.895 &   & -0.004 & 0.019 & 0.322 & 0.500 & 0.997 & 0.996 & 0.825 \\ \cline{2-18}
		&\multirow{3}{*}{$\log(3)$} & 250 & -0.019 & 0.026 & 0.220 & 0.307 & 0.915 & 0.938 & 0.830 &   & -0.005 & 0.043 & 0.458 & 0.642 & 0.992 & 0.986 & 0.743 \\
		&&500 & -0.011 & 0.014 & 0.163 & 0.222 & 0.941 & 0.943 & 0.869 &   & -0.003 & 0.032 & 0.397 & 0.572 & 0.995 & 0.991 & 0.776 \\
		&&1000 & -0.008 & 0.010 & 0.117 & 0.160 & 0.944 & 0.948 & 0.884 &   & -0.002 & 0.022 & 0.34 & 0.499 & 0.996 & 0.994 & 0.803 \\ \hline
		\multirow{9}{*}{\rotatebox[origin=c]{90}{None}} &
		\multirow{3}{*}{$-\log(3)$} & 250 & -0.005 & 0.016 & 0.198 & 0.295 & 0.928 & 0.915 & 0.946 &   & -0.002 & 0.05 & 0.419 & 0.628 & 0.994 & 0.987 & 0.900 \\
		&&500 & 0.001 & 0.005 & 0.147 & 0.215 & 0.938 & 0.931 & 0.98 &   & 0.003 & 0.032 & 0.376 & 0.551 & 0.995 & 0.991 & 0.918 \\
		&&1000 & 0.002 & 0.002 & 0.106 & 0.154 & 0.939 & 0.939 & 0.994 &   & 0.007 & 0.019 & 0.33 & 0.474 & 0.995 & 0.993 & 0.942 \\ \cline{2-18}
		&\multirow{3}{*}{$0$} & 250 & -0.003 & 0.011 & 0.198 & 0.274 & 0.933 & 0.929 & 0.961 &   & 0.004 & 0.030 & 0.440 & 0.609 & 0.996 & 0.992 & 0.877 \\
		&&500 & 0.002 & 0.003 & 0.146 & 0.198 & 0.945 & 0.938 & 0.988 &   & 0.007 & 0.021 & 0.388 & 0.538 & 0.996 & 0.993 & 0.917 \\
		&&1000 & 0.002 & 0.001 & 0.104 & 0.142 & 0.950 & 0.947 & 0.997 &   & 0.009 & 0.012 & 0.338 & 0.465 & 0.996 & 0.995 & 0.945 \\ \cline{2-18}
		&\multirow{3}{*}{$log(3)$} & 250 & -0.012 & 0.021 & 0.234 & 0.303 & 0.904 & 0.930 & 0.928 &   & 0.004 & 0.040 & 0.468 & 0.626 & 0.993 & 0.986 & 0.805 \\
		&&500 & -0.002 & 0.008 & 0.178 & 0.223 & 0.926 & 0.942 & 0.974 &   & 0.009 & 0.025 & 0.410 & 0.545 & 0.994 & 0.990 & 0.883 \\
		&&1000 & 0.001 & 0.002 & 0.130 & 0.161 & 0.943 & 0.945 & 0.992 &   & 0.012 & 0.012 & 0.358 & 0.473 & 0.995 & 0.992 & 0.931 \\ \hline
	\end{tabular}
\end{sidewaystable}	

With $q=5$ noise variables, the logistic model suffered substantially, inflating bias and interval width strongly  (Appendix \ref{sec:appendix_B}). Despite wider intervals, coverage was low even in large samples (approx. 0.90) and unacceptably low in small samples. To the contrary, for BART with $q=5$ the optimism bias was only slightly higher than under the correctly specified logistic model (Table \ref{tab:simres_OTRmax}). Furthermore, accuracy under strong heterogeneity was only slightly lower ranging from 0.926 ($n=250$, $\lambda=\log(3)$) to 0.959 ($n=1000$, $\lambda=\log(3)$). Under mild heterogeneity accuracy was still good ranging from 0.691 ($n=250$, $\lambda=\log(3)$) to 0.840 ($n=250$, $\lambda=0$); larger samples yielded a clear gain in accuracy here. Furthermore, BART had conservative coverage at 0.99 throughout conditions, where the correctly specified logistic regression model partly had too liberal coverage. However, credible intervals were wider on average than under the logistic model, emphasizing large samples are needed for precise inference with BART. Compared to the logistic model with $q=5$ (Appendix \ref{sec:appendix_B}), BART at $q=5$ had smaller bias and interval width at much better coverage. BART also strongly profited from data without noise variables in the set $\X$ ($q=0$; Appendix \ref{sec:appendix_B}). Interval width was then reduced (slightly wider than under the correctly specified logistic model), while coverage remained at high levels. Bias was slightly smaller and almost competitive with the correctly specified logistic model. Accuracy was increased throughout. 

We repeated this analysis for a second loss function, OTR.25 (Table  \ref{tab:simres_OTR25}). Throughout conditions, optimism bias under, both, the correctly specified model and BART was smaller. The methods also yielded very comparable coverage and width of credible intervals as under OTRmax loss. Likewise accuracy was similar for strong and mild treatment effect heterogeneity. Without heterogeneity, the optimal decision $a^{opt}_{\text{OTR.25}}$ was $0$, regardless of the value for $X_1$ (for OTR.25 loss function plots see Figure \ref{fig:true_loss} in Appendix \ref{sec:appendix_B}). This decision was recovered with good accuracy for both logistic regression and BART; as before larger samples yielded a substantial increase in accuracy (e.g., see BART with $\lambda=\log(3)$).

\section{OTR for oropharynx cancer}           \label{sec:application}
In this section we develop OTR for oropharynx cancer. We first explain details of the data, loss functions, and modelling. Subsequently, results on expected loss and survival probabilities under three candidate treatment regimes are presented.

\subsection{Data and model}
Observational data on patients diagnosed with Oropharyngeal Squamous Cell Carcinoma (OPSCC) at VUmc Amsterdam, The Netherlands, were available \parencite{nauta_evaluation_2018}. We optimized assignment of RT ($W=0$) vs. CRT ($W=1$) treatment using different loss functions. The proportion of patients receiving CRT was 55.2\%. We considered survival status at three years after treatment ($Y(0),Y(1)$). Patient covariates ($\X$) included age, gender, tumour T-stage (five categories, T1 to T4b), tumour N-stage (four categories, N0 to N3), alcohol use (unit years), smoking (package years), co-morbidity assessment (three categories, Adult Co-morbidity Evaluation-27), and human papillomavirus infection (HPV). The clinical decision for CRT ($W=1$) considers tumour stage, age, and co-morbidities. These confounders were observed in support of assumption (\ref{eq:ignorable_treatment}). To assess covariate overlap, we estimated propensity scores $P(W=1|\x_i;\hat{\bm{\beta}})$ for all $i=1,...,n$, from logistic model $P(W=1|\x_i;\bm{\beta})=\text{expit}(\x_i^T\bm{\beta})$ and compared the propensity distributions in groups $W=0,1$ (Appendix \ref{sec:appendix_C}). To facilitate overlap, patients aged 75 years or higher were omitted (n=277 after omission) and tumour stages N2 and N3 were merged. 

We considered three loss functions (regimes) in the conditional parametrization (Table \ref{tab:loss}): OTRmax with $(L^{(1)}_{00}, L^{(0)}_{01}, L^{(1)}_{10}, L^{(1)}_{11} ) = (0,1,1,0)$, OTR.25 with  $(0.25,1,1.25,0.25)$, and OTR.50 with  $(0.50,1,1.50,0.50)$. The coefficients $L^{(1)}_{00}$ and $L^{(1)}_{11}$ quantify loss due to unnecessary burden that patients with equal survival outcomes under both treatments receive from CRT assignment. These errors are penalised with 25\% and 50\% of the loss from a wrong treatment decision leading to death ($L^{(0)}_{01}$ scaled at 1). Coefficient $L^{(1)}_{10}$ is set to 1.25 (1.50) to reflect the loss due to wrong treatment by CRT and unnecessary burden simultaneously. We modelled the potential outcomes using BART with prior settings as discussed in section \ref{sec:simulation}. We multiply imputed missing data on the alcohol ($n=3$), comorbidity ($n=2$), and HPV variables ($n=7$) using predictive mean matching with five imputed data sets. BART models of the potential outcomes were fitted on each completed data set. This resulted in five posterior distributions of the marginal probabilities ($\pi_{1+},\pi_{+1}$) with $5 \times 10^3$ draws respectively. Following \textcite{gelman_bayesian_2013}, these posterior distributions were pooled for inference.

\subsection{Results}
We first held the partial association of potential outcomes fixed at $\phi_0=1$ and used posterior mean decision rule $a_1^*$ (\ref{eq:OTR_rule1}) to assess the OTR. We found that 60.3\% of patients received CRT under OTRmax loss, but only 13.4\% under OTR.25 and 0.7\% under OTR.50 caused by increased penalization of unnecessary burden in these regimes. The empirical distribution of the posterior probabilities for assigning CRT correctly, $\rho(\x_i)$ in (\ref{eq:pcp}), is shown in Figure \ref{fig:posterior_prob}. The decision boundary for $a^*_2=1$ (\ref{eq:OTR_rule2}) at $\rho=0.5$ is illustrated. The OTRmax probabilities at each quantile were higher than for OTR.25 and OTR.50, reflecting that more patients benefit from CRT under OTRmax loss. We carried out a sensitivity analysis for $\phi$ by evaluating all decisions again at the conservative bounds $\phi=[\exp(-3),\exp(3)]$. For $a_1^*$, 4.7\% (OTR.25) and 1.4\% (OTR.50) of all decisions changed between these bounds; these decisions are called sensitive to $\phi$. The impact of $\phi$ on $\rho(\x_i)$ is illustrated by bounds in Figure \ref{fig:posterior_prob}. Only few bounds enclosed the decision boundary at 0.5 demonstrating non-sensitivity of most decisions. In addition, we used black/grey coding in Figure \ref{fig:posterior_prob} to compare the posterior probabilities to treatments received ($W$). At the left tail of the distribution of $\rho(\X)$ under OTRmax more patients received RT ($W=0$) than CRT; similarly more patients at the right tail received CRT ($W=1$) than RT. Averaged across the sample, 60.3\% (OTRmax), 51.6\% (OTR.25), and 45.5\% (OTR.50) received the same treatment under optimal assignment ($a_1^*$) and the one observed ($W$). OTRmax had strongest overlap with the observed regime, but there were still substantial differences in assignment strategies. 

\begin{figure}[h!]  
	\centering
	\begin{subfigure}{0.32\textwidth}
		\includegraphics[trim=140 0 150 0, clip=true, width=\linewidth]{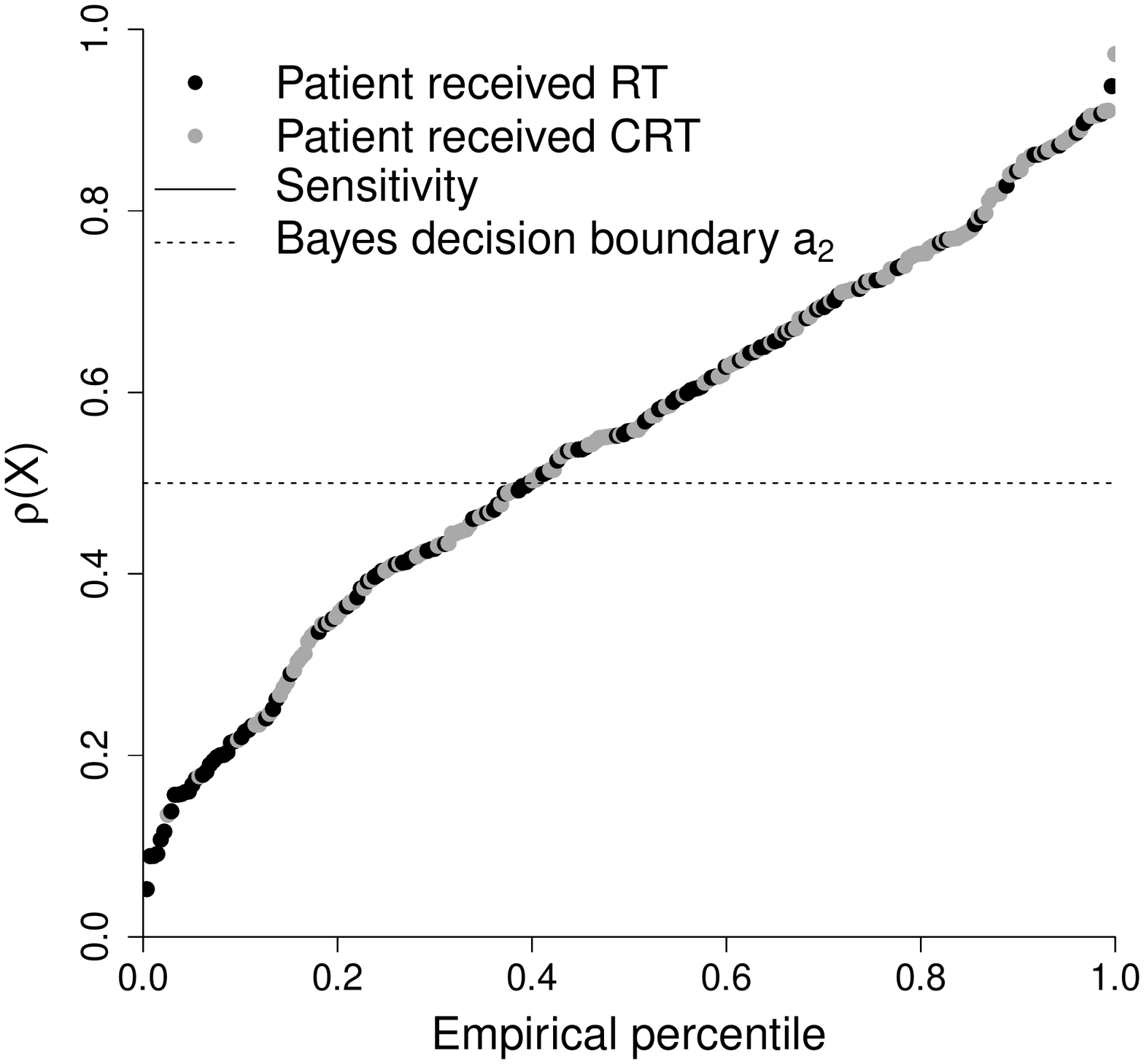}
		\caption{OTRmax} \label{fig:fig_post_prob1}
	\end{subfigure} 
	\begin{subfigure}{0.32\textwidth}
		\includegraphics[trim=140 0 150 0, clip=true, width=\linewidth]{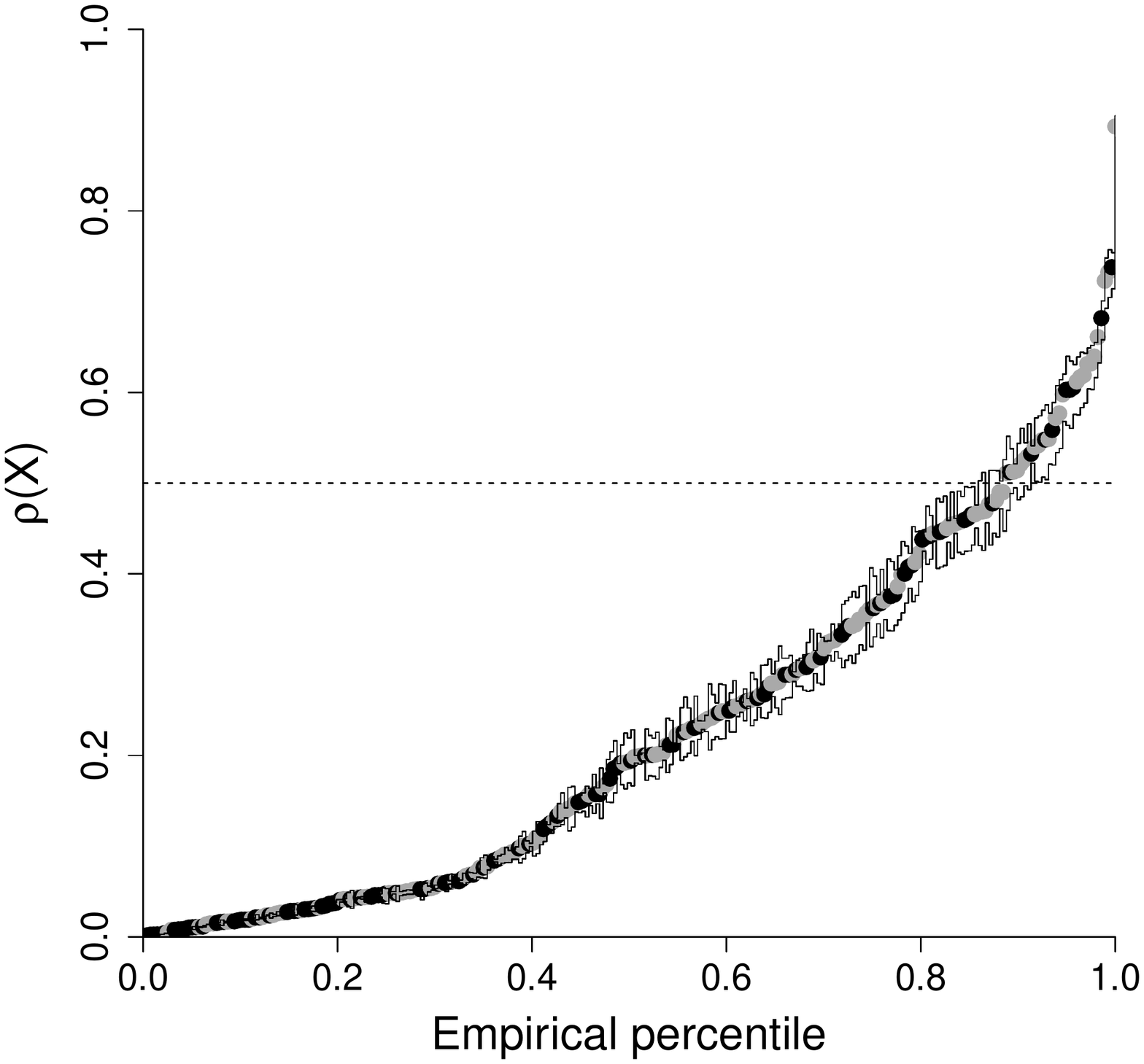}
		\caption{OTR.25} \label{fig:fig_post_prob2}
	\end{subfigure} 
	\begin{subfigure}{0.32\textwidth}
		\includegraphics[trim=140 0 150 0, clip=true, width=\linewidth]{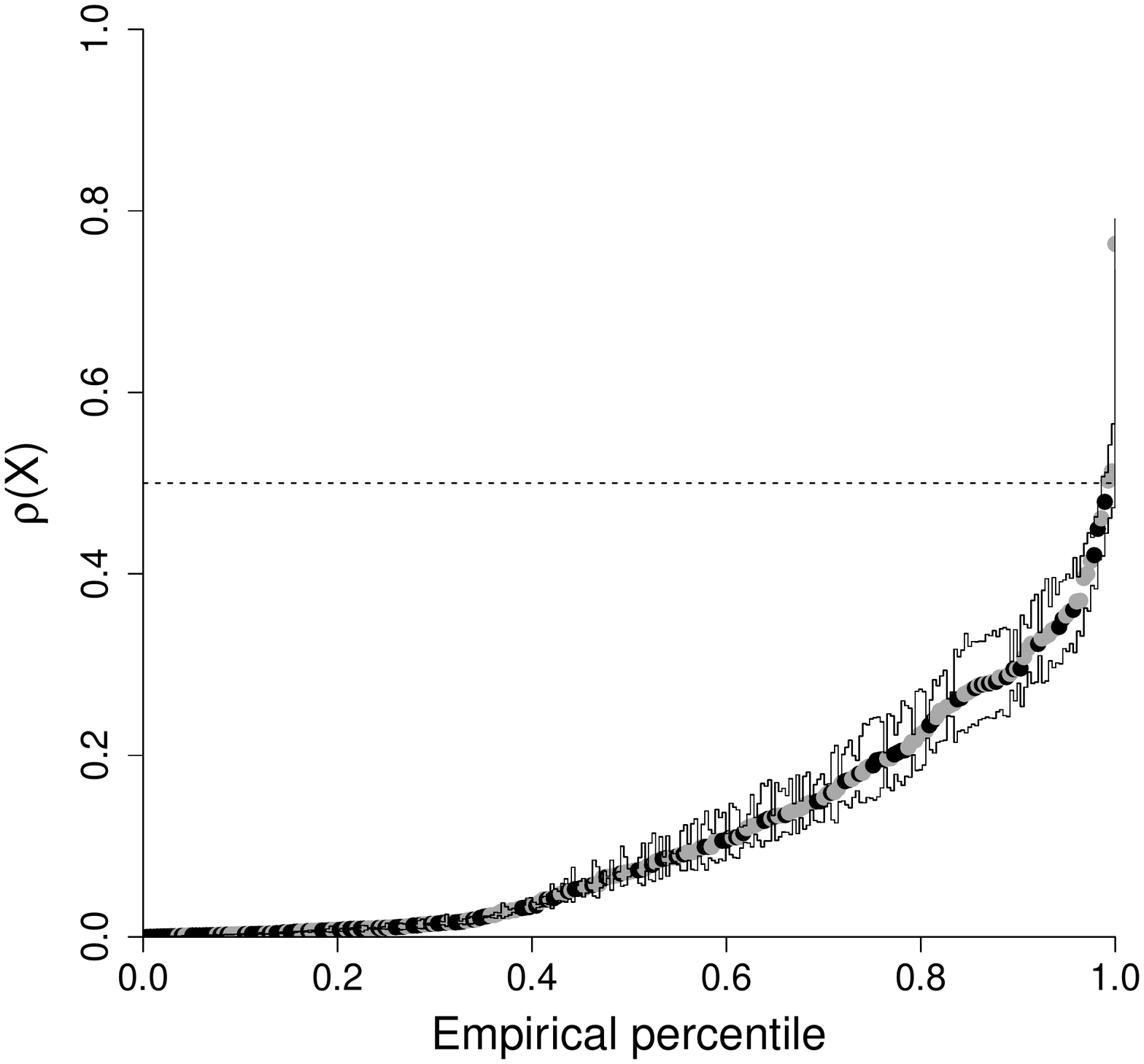}
		\caption{OTR.50} \label{fig:fig_post_prob3}
	\end{subfigure} 
%
	\caption{Empirical quantile plots of posterior probabilities for correct CRT assignment at $\phi_0=1$ for different loss functions (regimes). A sensitivity analysis is carried out for $\log(\phi) \in \{-3,3\} $.} \label{fig:posterior_prob}
\end{figure}

The survival probabilities under the OTR, i.e. $\mu_Y(a_1^*(\x_i),\bm{\theta},\x_i)$ in (\ref{eq:mu_Y}), were compared to probabilities under the observed regime $W$, $\mu_Y(w_i,\bm{\theta},\x_i)$  (Figure \ref{fig:quantile_survival}).  A sensitivity analysis is not needed for this estimand (section \ref{sec:sensitivity}). At any quantile the distribution under OTRmax had slightly higher value than under the observed regime $W$. Largest differences were present at the lower tail suggesting most potential for survival probability optimization concerns patients with small survival probabilities under both treatments (differences in survival probabilities amount to 0.3). Uncertainty around posterior means is indicated by 95\% credible intervals, where we found lower probabilities had higher posterior uncertainty than higher probabilities. For OTR.25 and OTR.50 uncertainty was similar but the gain in survival was nullified. This reflects the objective of loss minimization instead of survival probability maximization. This trade-off is further illustrated by comparing $\mu_Y(a_1^*(\x_i),\bm{\theta},\x_i)$ to expected loss $\mu_{\Loss}(a_1^*(\x_i),\bm{\theta},\x_i)$, (\ref{eq:exp_loss}), in Figure \ref{fig:quantile_loss}. Expected loss of all regimes was lower than that of the observed regime and for OTR.25 and OTR.50 the reduction in loss relative to the observed regime was substantial. 

\begin{figure}[h!]  
	\centering
	\begin{subfigure}{0.32\textwidth}
		\includegraphics[trim=140 0 150 0, clip=true, width=\linewidth]{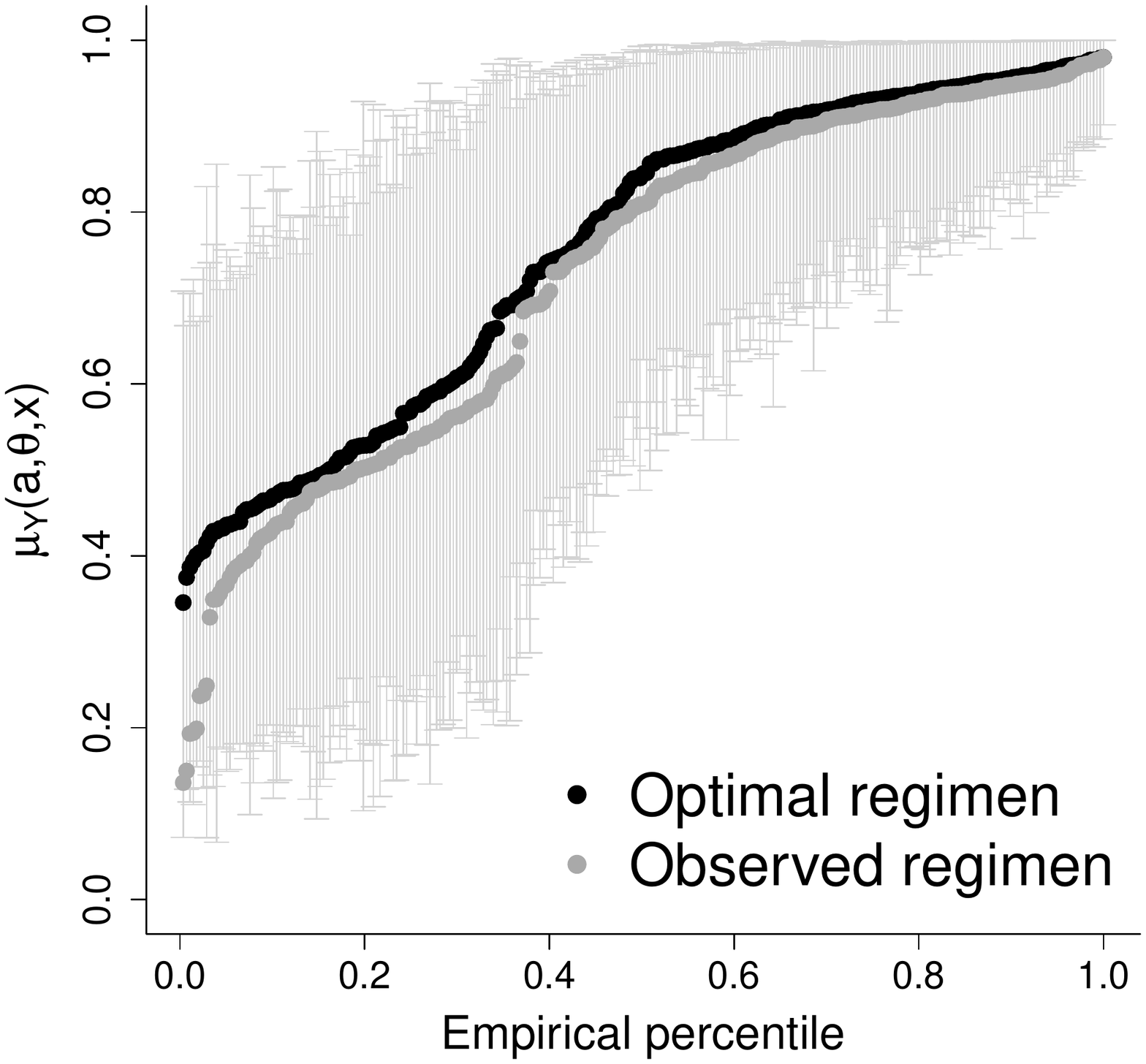}
		\caption{OTRmax} \label{fig:quantile_survival1}
	\end{subfigure} 
	\begin{subfigure}{0.32\textwidth}
		\includegraphics[trim=140 0 150 0, clip=true, width=\linewidth]{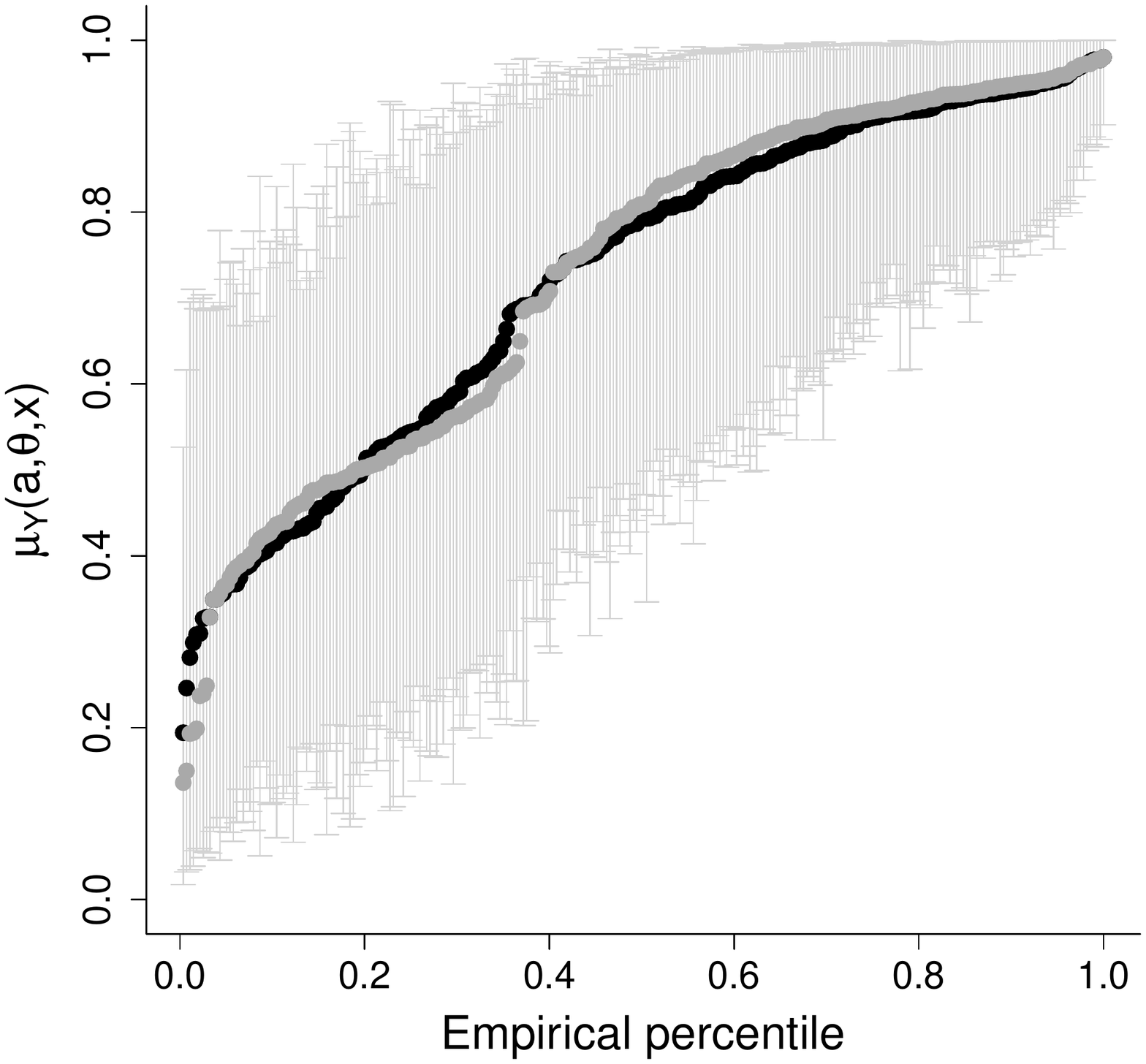}
		\caption{OTR.25} \label{fig:quantile_survival2}
	\end{subfigure}
	\begin{subfigure}{0.32\textwidth}
		\includegraphics[trim=140 0 150 0, clip=true, width=\linewidth]{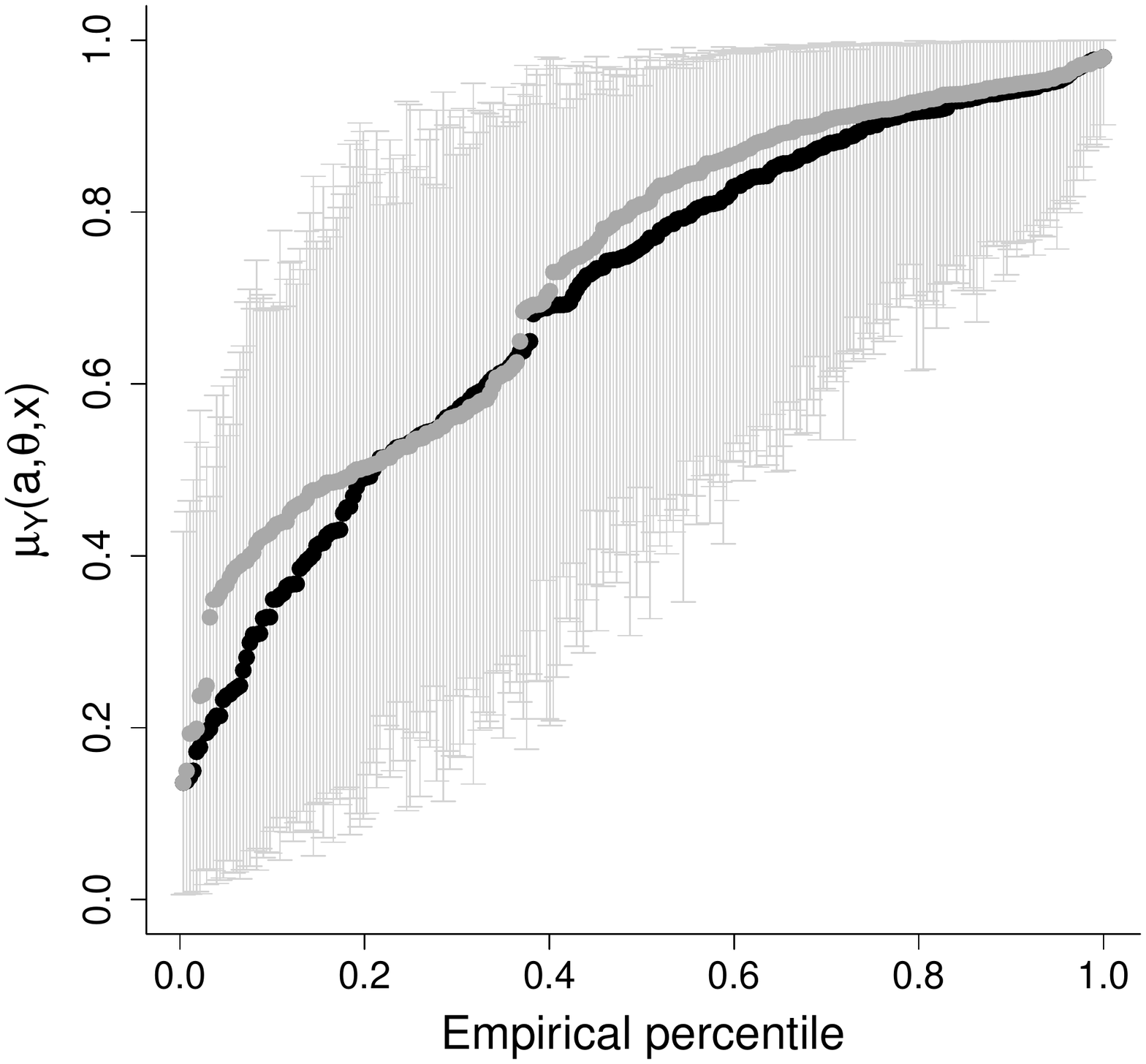}
		\caption{OTR.50} \label{fig:quantile_survival3}
	\end{subfigure} 
	\caption{Mean posterior survival probabilities by OTR $\mu_Y(a_1^*(\x_i),\bm{\theta},\x_i)$ as compared to the observed regime $\mu_Y(w_i,\theta,\x_i)$. 95\% credible intervals given for $\mu_Y(a_1^*(\x_i),\bm{\theta},\x_i)$. The figure includes all $n=277$ patients, i.e. also patients with sensitive decisions for which we fixed $\phi=1$. An alternative is to omit these patients as shown in Appendix \ref{sec:appendix_C} with similar results.} \label{fig:quantile_survival}
\end{figure}

\begin{figure}[h!] 
	\centering
	\begin{subfigure}{0.32\textwidth}
		\includegraphics[trim=140 0 150 0, clip=true, width=\linewidth]{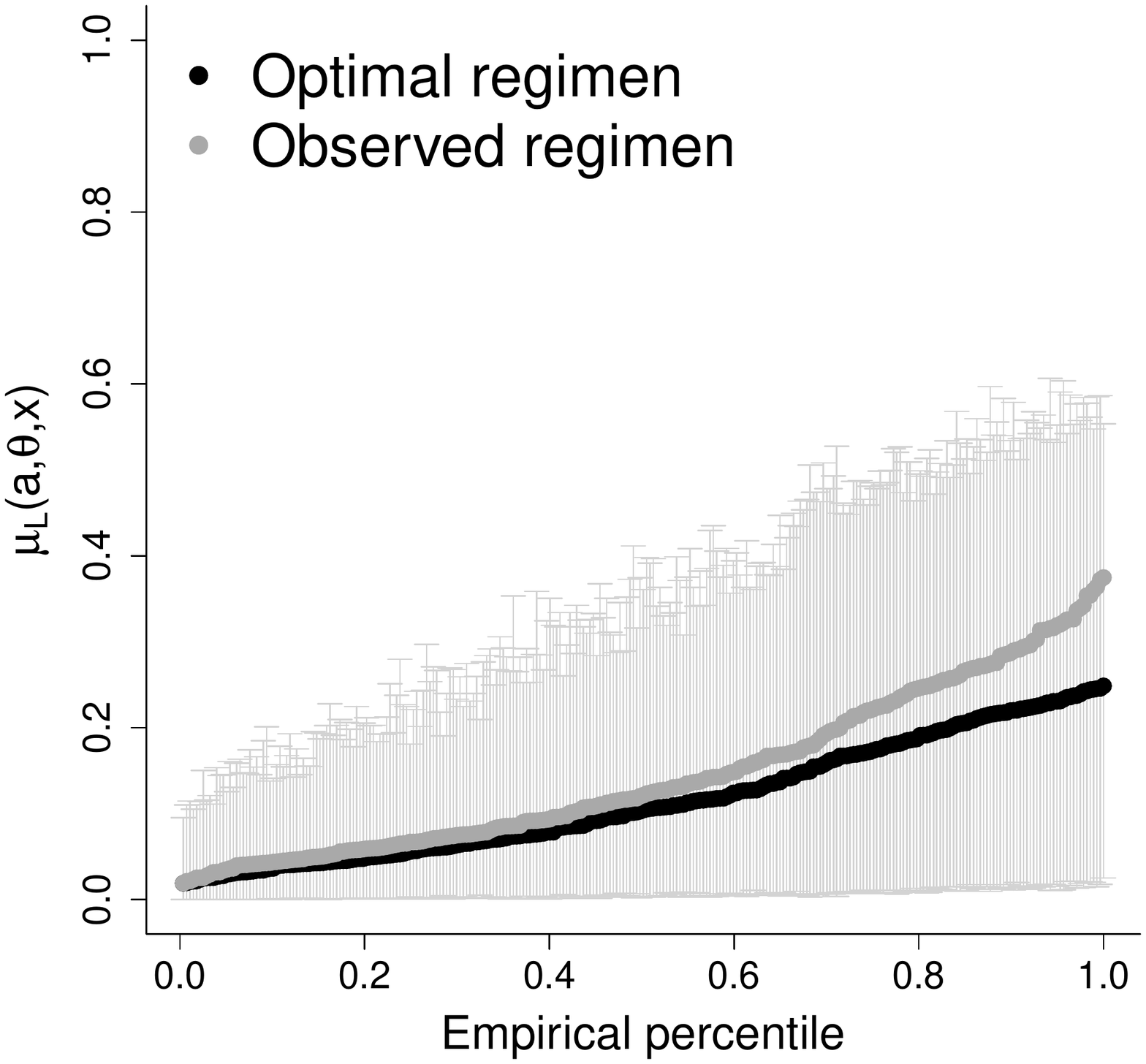}
		\caption{OTRmax} \label{fig:quantile_loss1}
	\end{subfigure}  
	\begin{subfigure}{0.32\textwidth}
		\includegraphics[trim=140 0 150 0, clip=true, width=\linewidth]{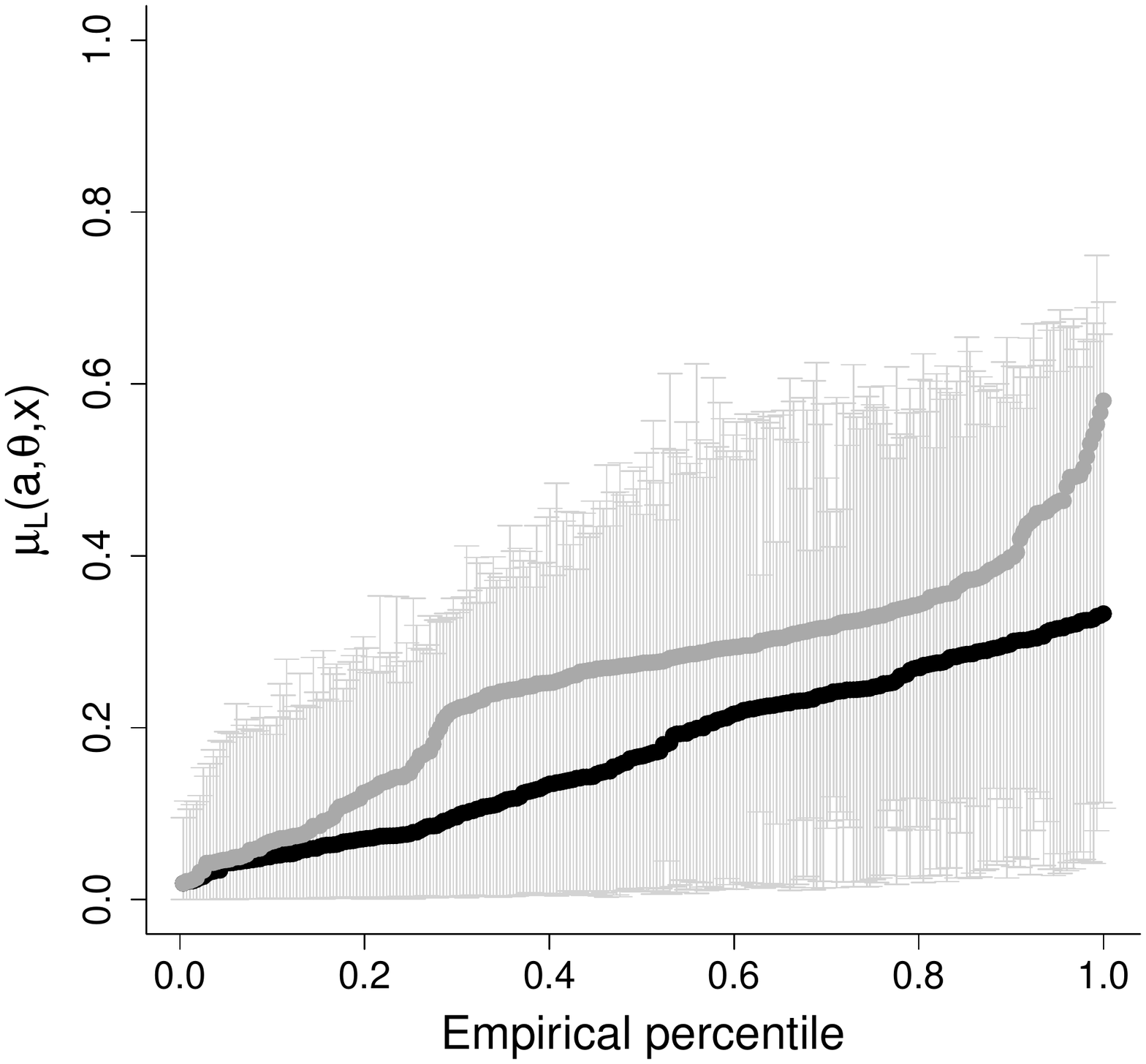}
		\caption{OTR.25} \label{fig:quantile_loss2}
	\end{subfigure}
	\medskip
	\begin{subfigure}{0.32\textwidth}
		\includegraphics[trim=140 0 150 0, clip=true, width=\linewidth]{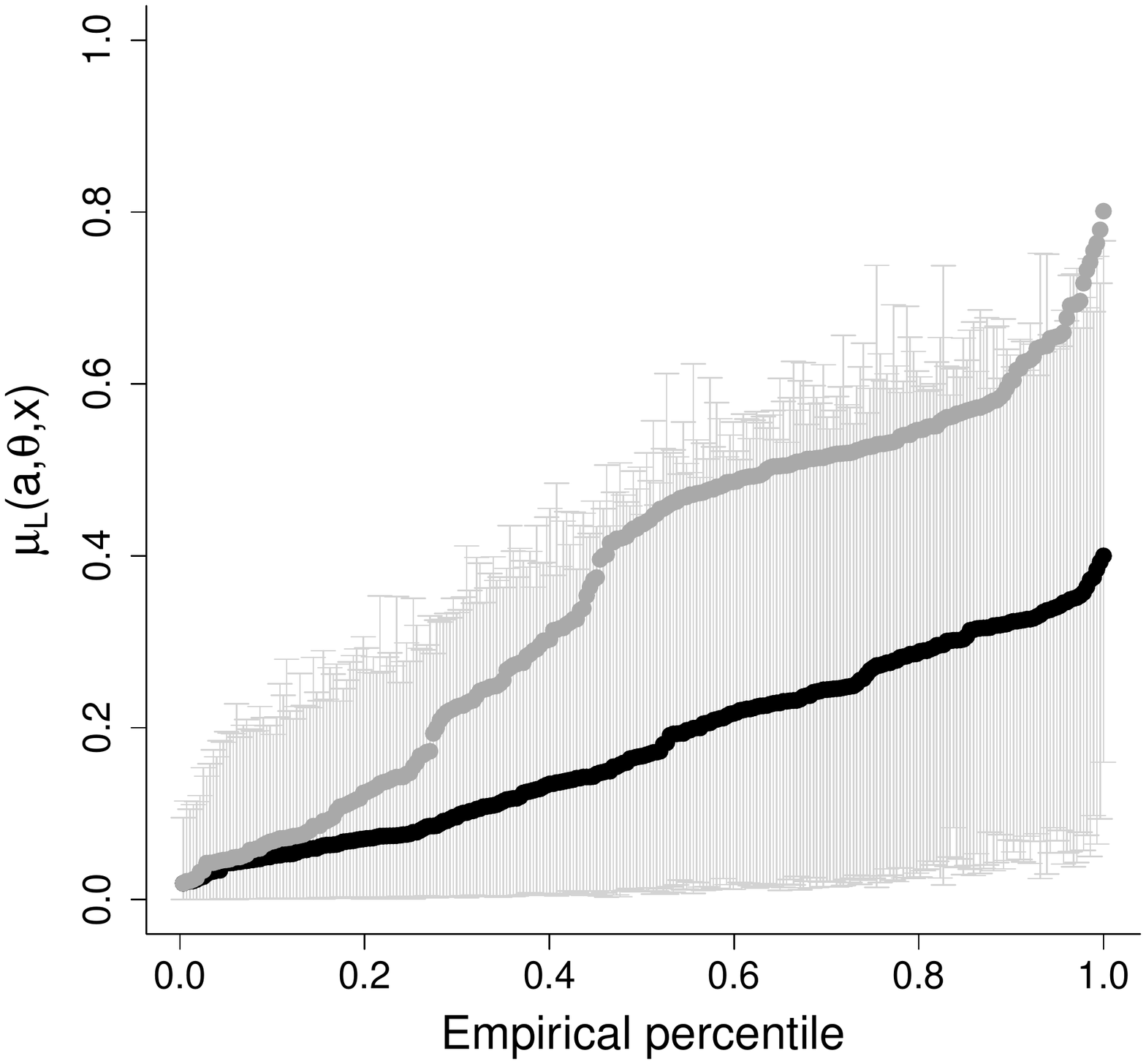}
		\caption{OTR.50} \label{fig:quantile_loss3}
	\end{subfigure}
	\caption{ Mean posterior expected loss $\mu_{\Loss}(a_1^*(\x_i),\bm{\theta},\x_i)$ under OTR (optimal regime) and the observed regime $\mu_{\Loss}(w_i,\bm{\theta},\x_i)$. Credible intervals given for OTR. Parameter uncertainty of $\phi$ simulated by prior $\phi \sim \text{U}(\exp(-3), \exp(3))$.}
	\label{fig:quantile_loss}
\end{figure}

There is additional uncertainty in posterior estimates of $\mu_{\Loss}(a_1^*(\x_i),\bm{\theta},\x_i)$ due to sensitivity of inference to association $\phi$ (section \ref{sec:sensitivity}). Keeping decisions $a_1^*(\x_i)$ fixed at $\phi_0=1$, we assumed a uniform prior on $\phi \sim \text{U}(\exp(-3), \exp(3))$ and integrated $\phi$ out during the estimation phase. The interval width in Figure \ref{fig:quantile_loss} was thereby increased compared to e.g. $\phi=1$ (conditional independence of potential outcomes), in particular by 7.0\% (OTRmax), 5.4\% (OTR.25), and 5.9\% (OTR.50). These results demonstrate that intervals are robust to strong prior uncertainty about $\phi$.

To summarize results, we average survival probabilities and expected losses (Table \ref{tab:result_summary}), where we define $\mathbb{U}(a) := n^{-1} \sum_i \hat{\mu}(a_i,\bm{\theta},\x_i)$, with $\hat{\mu}$ the posterior mean estimates of $\mu_Y$ or $\mu_{\Loss}$ (shown in Figures \ref{fig:quantile_survival} and \ref{fig:quantile_loss}); furthermore, $\bar{W}= n^{-1} \sum_i W_i $ the sample proportion of patients receiving CRT under observed treatment assignment $W$ and $\bar{ a}_1^*= n^{-1} \sum_i a_{1i}^*(\x_i)$ the sample proportion of patients receiving CRT under the OTR. The average expected loss under the observed regime increased strongly from OTRmax to OTR.50 loss functions, but remained on low levels under the optimal regime $a_1^*$. Average survival probability under $W$ was 72.8\%. OTRmax loss increased average survival probability to 75.9\%, whereas OTR.25 and OTR.50 had averaged probabilities only slightly below those of the observed regime. However, CRT assignment could be strongly reduced, in particular below the observed regime (55.2\%). In conclusion, the decision which loss function and regime are 'best' remains subjective. However, a good candidate is OTR.25 which reduced patient burden by strongly reducing CRT assignment by four times (55.2 to 13.4\%) without lowering the average survival probability much under that of the observed regime (72.8 vs. 72.5\%; emphasis added in Table \ref{tab:result_summary}). 

\begin{table}[h!]
	\caption{Sample-averaged posterior mean estimates of expected losses and survival probabilities under the OTR ($a_1^*$) and observed regime ($W$) and the proportion of patients assigned to CRT. }
	\label{tab:result_summary}
	\centering
	\begin{tabular}{lcccccccc}
	& \multicolumn{2}{c}{Average loss} && \multicolumn{2}{c}{Average survival prob.} && \multicolumn{2}{c}{Prop. assigned CRT} \\
	\cline{2-3}
	\cline{5-6}
	\cline{8-9}
	\\[-1em]
	Loss function & $\mathbb{U}_{\Loss}(W)$ & $\mathbb{U}_{\Loss}(a_1^*)$ && $\mathbb{U}_{Y}(W)$ & $\mathbb{U}_{Y}(a_1^*)$ && $\bar{W}$ & $ \bar{ a}_1^*$ \\ \hline
	OTRmax & 0.148 & 0.118 &   & 0.728 & 0.759 &   & 0.552 & 0.603 \\
	OTR.25 & \textbf{0.259} & \textbf{0.173} &   & \textbf{0.728} &\textbf{ 0.725} &   & \textbf{0.552} & \textbf{0.134} \\
	OTR.50 & 0.370 & 0.181 &   & 0.728 & 0.697 &   & 0.552 & 0.007 \\ \hline
	\end{tabular}
\end{table}

\section{Discussion}            \label{sec:discussion}
We suggested a decision theoretical framework for estimating optimal treatment regimes (OTR) based on a loss function specified on the bivariate distribution of dichotomous potential outcomes. Our approach distinguishes expected patient loss from expected outcomes (e.g., survival probabilities), so that the effects of minimized loss can be compared and traded off against the expected outcomes under an OTR (Table \ref{tab:result_summary}). Furthermore, the particular loss parametrization can be flexibly tailored to the application, where we advocated using the conditional parametrization which is closely aligned to clinical decision making. Importantly, it allows penalizing unnecessary treatment burden (besides treatment failure), which is a central consideration in clinical practice. In an application to oropharyngeal cancer (OPSCC), we could achieve a fourfold reduction in burdensome chemotherapy (CRT) assignment without lowering average survival probabilities below the level of the observed regime. Posterior probabilities for treatment assignment as well as credible intervals and posterior mean estimates for expected loss and outcomes under an OTR are strongly informative measures that are useful to support clinicians in decision making. For the OPSCC data, uncertainty in expected loss and survival was larger at the lower tail of the distributions (Figures \ref{fig:quantile_survival2} and \ref{fig:quantile_loss2}), but posterior certainty was high that radiotherapy without chemotherapy was the optimal strategy for many patients (Figure \ref{fig:fig_post_prob2}). 

Furthermore, we developed Bayesian OTR estimation methodology assuming ignorable treatment assignment. An advantage of our approach is that it relies on two marginal Bayesian models of the potential outcomes, which allows applying standard Bayesian techniques including model updating, model and variable selection, and inclusion of prior information. In practice, however, the functional relationship between patient covariates and the potential outcomes is often unknown. Bayesian statistical and machine learning may then facilitate functional approximation, where we evaluated and applied Bayesian additive regression trees, BART. Our estimation approach is related to variants suggested by \textcite{zajonc_bayesian_2012} and \textcite{murray_bayesian_2017} who use the posterior predictive distribution in the context of dynamic regimes, whereas we base inference on the posterior distribution of the conditional marginal probabilities of the potential outcomes and a prior assumption on their partial association. An important assumption of our approach is ignorable treatment assignment requiring that all confounders are observed. Thorough planning using expert knowledge is essential to have data on the required confounders. For the case of the OPSCC data, all known causes of CRT assignment, by expert knowledge and treatment protocols, were observed in support of the assumption.

In simulations we found an optimism bias of point estimates under many conditions where we varied loss functions, treatment assignment mechanisms, treatment effect heterogeneity, and sample size. In particular, the direction of bias of expected loss was negative (and positive for outcomes) in small samples with low or no heterogeneity between expected losses. Bias and interval width decreased and decision accuracy increased substantially in larger samples. Frequentist interval coverage probabilities were nevertheless at nominal level for large samples in the Bayesian logistic model, and conservative for BART, even in small samples. Nevertheless, large samples appear desirable to control optimism bias and variance, needed in particular for BART. Replication of our findings on the optimal OPSCC regime, based on $n=277$, in larger samples is therefore desirable.

Several further paths for future research are worth noting. Our approach may be adapted to let loss functions depend on patient characteristics or preferences, which could broaden its usefulness for personalized medicine. Furthermore, it is desirable to extend the approach to loss functions and estimation for time-to-event (survival) outcomes and more than two treatments.
\vspace{2cm}

\newpage
\printbibliography
\newpage

\begin{appendices}
\setcounter{table}{0}
\setcounter{figure}{0}
\renewcommand{\thetable}{A\arabic{table}}
\renewcommand{\thefigure}{A\arabic{figure}}

\section{Proofs} \label{sec:appendix_A}

\begin{proof}[Proof of Theorem \ref{thm:bonferroni}.]
	The coverage of $I(a^*)$ is 
	\begin{align*}
	P(\mu_{\Loss}(a^*,\bm{\theta}_0,\x) \in I(a^*)) 
	&= P(\mu_{\Loss}(1,\bm{\theta}_0,\x) \in I(1), a^* = 1) + P(\mu_{\Loss}(0,\bm{\theta}_0,\x) \in I(0),a^*=0) \\
	&\geq P(a^*=1) - \alpha + P(a^*=0) - \alpha = 1-2\alpha.\qedhere
	\end{align*}
\end{proof}

The remainder of this section is devoted to the proof of Theorem \ref{thmcredible}. We consider (non-Bayesian) confidence intervals and prove Theorem \ref{thmcoverage} below. From this theorem we immediately obtain Theorem \ref{thmcredible}.

Let $(X,Y)$ be a two-dimensional random vector having the bivariate normal distribution with unknown mean $(\mu,\nu)$ and positive-definite covariance matrix 
\[
\left( \begin{array}{cc} \sigma^2 & \rho \sigma \tau \\ \rho\sigma\tau & \tau^2 \end{array} \right),
\]
with $\sigma^2$ and $\tau^2$ known. The classical confidence intervals for $\mu$ and $\nu$ with coverage $1-\alpha$ are given by
\begin{equation}\label{confidenceintervals}
[ X - \sigma \xi_{1-\alpha/2}, X + \sigma \xi_{1-\alpha/2} ],\qquad [ Y - \tau \xi_{1-\alpha/2}, Y + \tau \xi_{1-\alpha/2} ],
\end{equation}
with $\xi_{1-\alpha/2}$ the $(1-\alpha/2)$ quantile of the standard normal distribution. We now define a decision-based confidence interval $I$ by selecting one of the two intervals based on the outcome of $(X,Y)$, as follows:
\[
I = \left\{ \begin{array}{ll} 
[ X - \sigma \xi_{1-\alpha/2}, X + \sigma \xi_{1-\alpha/2} ], & \mbox{if } X \leq Y, \\
\normalsize[ Y - \tau \xi_{1-\alpha/2}, Y + \tau \xi_{1-\alpha/2} ], & \mbox{if } X > Y.  
\end{array} \right.
\]
I.e.\ we select the interval with the smallest mean. We define the \emph{coverage} of $I$ as $P(T\in I)$, with
\[
T = \left\{ \begin{array}{ll}
\mu, & \mbox{if } X \leq Y, \\
\nu, & \mbox{if } X> Y.
\end{array} \right.
\]

As an appetizer, suppose $\mu=\nu=0$, $\sigma^2 = \tau^2=1$ and $\rho =0$. Then
\[
I = [ \min\{X,Y\} - \xi_{1-\alpha/2}, \min\{X,Y\} + \xi_{1-\alpha/2} ]
\]  
and the coverage of $I$ is 
\begin{align*}
P(0\in I) & = P(0 < \min\{X,Y\} + \xi_{1-\alpha/2}) - P(0< \min\{X,Y\} - \xi_{1-\alpha/2}) \\
& = (1-\alpha/2)^2 - (\alpha/2)^2 = 1-\alpha.
\end{align*}
Hence the interval $I$ has a negative bias, but nevertheless the coverage is $1-\alpha$.

\begin{theorem}\label{thmcoverage}
	The coverage of $I$ is
	\begin{itemize}
		\item[(i)] equal to $1-\alpha$, if $\mu=\nu$ or $\sigma^2 = \tau^2$,
		\item[(ii)] strictly between $1-\alpha$ and $1-\frac{1}{2}\alpha$, if $\mu>\nu$ and $\sigma^2<\tau^2$,
		\item[(iii)] strictly between $1-\frac{3}{2}\alpha$ and $1-\alpha$, if $\mu>\nu$ and $\sigma^2 > \tau^2$.
	\end{itemize}
\end{theorem}

Before proceeding to the proof of Theorem \ref{thmcoverage}, we show how Theorem \ref{thmcredible} follows from it. By the assumptions in Theorem \ref{thmcredible}, the credible intervals $I(1), I(0)$ are of the form \eqref{confidenceintervals}. Indeed, $X,Y$ represent the means of the intervals, which by assumption have the bivariate normal distribution with variances $\sigma^2 = \sigma^2(1)$ and $\tau^2 = \sigma^2(0)$. By normality of the means, by the assumption that the widths of the intervals are constant, and because we assumed the coverage of $I(1),I(0)$ to be symmetric in the sense of \eqref{coverageI1}, we must have that $\mu = \mu_{\Loss}(1,\bm{\theta}_0,\x)$, $\nu = \mu_{\Loss}(0,\bm{\theta}_0,\x)$ and
\[
I(1) = [X - \sigma \xi_{1-\alpha/2}, X + \sigma \xi_{1-\alpha/2}],\qquad
I(0) = [Y - \tau \xi_{1-\alpha/2}, Y + \tau \xi_{1-\alpha/2}].
\]
This proves Theorem \ref{thmcredible}; it remains to prove Theorem \ref{thmcoverage}.

\begin{proof}[Proof of Theorem \ref{thmcoverage}.] We write $\xi=\xi_{1-\alpha/2}$ and denote
	\[
	\tilde x = \frac{x-\mu}{\sigma},\qquad \tilde y = \frac{y-\nu}{\tau}.
	\]
	We can write the coverage of $I$ as
	\begin{align}
	& P(X > Y, \nu \in [ Y - \tau \xi, Y + \tau \xi] ) \nonumber\\
	& \quad + P( X < Y, \mu \in [ X - \sigma \xi, X + \sigma \xi] ) \nonumber\\
	& = P(X> Y, \tilde X > \xi, | \tilde Y | < \xi ) \label{2a}\\
	& \quad + P(X>Y, |\tilde X|  < \xi, |\tilde Y| < \xi ) \label{2b}\\
	& \quad + P(X>Y, \tilde X < -\xi, |\tilde Y|  < \xi ) \label{2c}\\
	& \quad + P(X<Y, |\tilde X|  < \xi, \tilde Y > \xi ) \label{1a}\\
	& \quad + P(X<Y, |\tilde X| < \xi, |\tilde Y| < \xi ) \label{1b}\\
	& \quad + P(X<Y, |\tilde X| < \xi, \tilde Y < -\xi ). \label{1c}
	\end{align}
	Note that we must have $\eqref{2c} = 0$ or $\eqref{1c} = 0$, which can be easily seen from a two-dimensional plot of $\tilde y$ against $\tilde x$. This implies that the coverage of $I$ is less than $P(\tilde X > -\xi)$ or $P(\tilde Y > -\xi)$, respectively, which gives the upper bound of statement (ii) of the theorem. Furthermore, 
	\[
	\eqref{2b} + \eqref{1b} = P(|\tilde X| < \xi, |\tilde Y| < \xi)
	\]
	and 
	\[
	\eqref{2a} = P(\tilde X > \xi, |\tilde Y| < \xi) \quad \text{or} \quad 
	\eqref{1a} = P(|\tilde X| < \xi, \tilde Y > \xi).
	\]
	Therefore, the coverage of $I$ is greater than
	\[
	1-P(\tilde X < -\xi) - P(|\tilde Y| > \xi) \quad \text{or} \quad
	1-P(|\tilde X| > \xi) - P(\tilde Y < -\xi),
	\]
	respectively. This yields the lower bound of statement (iii).

	Using the decomposition \eqref{2a}--\eqref{1c}, we will prove that the coverage of $I$ is (i) equal to, (ii) greater than, or (iii) less than
	\begin{equation}\label{toprove}
	P( \tilde X > -\xi, \tilde Y > -\xi )
	- P( \tilde X > \xi, \tilde Y > \xi ),
	\end{equation}
	if (i) $\mu=\nu$ or $\sigma^2=\tau^2$, (ii) $\mu>\nu$ and $\sigma^2<\tau^2$, or (iii) $\mu>\nu$ and $\sigma^2>\tau^2$, respectively. The theorem then follows from the fact that \eqref{toprove} equals
	\begin{align*}
	& 1 - P(\tilde X < -\xi) - P(\tilde Y < -\xi) \\
	&\quad + P( \tilde X < -\xi, \tilde Y < -\xi ) 
	- P( \tilde X > \xi, \tilde Y > \xi ) \\
	& = 1 - \alpha/2 - \alpha/2 + 0 \\
	&= 1 - \alpha,
	\end{align*}
	where we use the fact that $(\tilde X,\tilde Y)$ has the same distribution as $(-\tilde X,-\tilde Y)$. It remains to prove statement \eqref{toprove}; we split the proof in four cases.

	\textbf{Case 1.} Suppose $\mu\geq \nu$, $\sigma^2\leq \tau^2$ and 
	\begin{equation}\label{case100}
	\mu-\nu \leq (\tau-\sigma) \xi.
	\end{equation}
	We will evaluate each of the terms \eqref{2a}--\eqref{1c}, and show that their sum is equal to \eqref{toprove} if $\mu=\nu$, and greater than \eqref{toprove} if $\mu>\nu$. Obviously,
	\begin{equation}\label{case10}
	\eqref{2b} + \eqref{1b} = P(|\tilde X| < \xi, |\tilde Y| < \xi ).
	\end{equation}
	Furthermore, if $\tilde x < \xi$ and $\tilde y  > \xi$ then, by assumption \eqref{case100},
	\[
	y-x = \tau \tilde y +\nu - \sigma \tilde x -\mu
	> (\tau-\sigma) \xi - (\mu-\nu) \geq 0,
	\] 
	which implies that
	\begin{equation}\label{case11}
	\eqref{1a} = P(|\tilde X| < \xi, \tilde Y > \xi).
	\end{equation}
	Similarly, if $\tilde x > -\xi$ and $\tilde y < -\xi$ then
	\[
	y-x = \tau \tilde y +\nu - \sigma \tilde x -\mu
	< - (\tau-\sigma) \xi - (\mu-\nu) \leq 0,
	\]
	which gives
	\begin{equation}\label{case12}
	\eqref{1c} = 0.
	\end{equation}
	Using the fact that $(\tilde X,\tilde Y)$ has the same distribution as $(-\tilde X,-\tilde Y)$, we obtain
	\begin{align*}
	\eqref{2c} &= P( \tilde X > (\tau/\sigma) \tilde Y - (\mu-\nu)/\sigma, \tilde X < -\xi, | \tilde Y| < \xi ) \\
	&\geq P( \tilde X > (\tau/\sigma) \tilde Y + (\mu-\nu)/\sigma, \tilde X < -\xi, | \tilde Y | < \xi ) \\
	&= P( - \tilde X > (\tau/\sigma)(-\tilde Y) + (\mu-\nu)/\sigma, - \tilde X < -\xi, | -\tilde Y | < \xi ) \\
	&= P( X < Y, \tilde X > \xi, | \tilde Y| < \xi ), 
	\end{align*}
	where the inequality is in fact an equality if and only if $\mu=\nu$. It follows that
	\begin{equation}\label{case13}
	\eqref{2a} + \eqref{2c} \geq P(\tilde X > \xi, |\tilde Y| < \xi).
	\end{equation}
	Combining \eqref{case10}, \eqref{case11}, \eqref{case12} and \eqref{case13} yields that the coverage of $I$ is greater than or equal to \eqref{toprove}, with equality if and only if $\mu=\nu$.

	\textbf{Case 2.} Suppose $\mu\geq \nu$, $\sigma^2\leq \tau^2$ and $\mu-\nu \geq (\tau-\sigma) \xi$. To give a bound on \eqref{2c}, consider the two lines
	\begin{align}
	\tilde y &= \frac{\sigma}{\tau} \tilde x + \frac{\mu-\nu}{\tau}, \label{line1} \\
	\tilde y &= \frac{\tau}{\sigma} \tilde x + \frac{\mu-\nu}{\sigma}. \label{line2}
	\end{align}
	The two lines intersect at the point $(-(\mu-\nu) / (\sigma+\tau), (\mu-\nu) / (\sigma + \tau))$, which is outside the region $\tilde x < -\xi, \tilde y < \xi$. If $\sigma^2 < \tau^2$ then the slope of line \eqref{line2} is greater than the slope of line \eqref{line1}. Therefore, if 
	\[
	\tilde x < -\xi,\qquad \tilde y < \xi, \qquad \tilde y < \frac{\tau}{\sigma} \tilde x + \frac{\mu-\nu}{\sigma},
	\]
	then we must have $\tilde y < (\sigma/\tau) \tilde x + (\mu-\nu) / \tau$. This implies
	\begin{align}
	\eqref{2c} &= P( \tilde Y < (\sigma/\tau) \tilde X + (\mu-\nu)/\tau, \tilde X < -\xi, | \tilde Y| < \xi ) \nonumber\\
	&\geq P( \tilde Y < (\tau/\sigma) \tilde X + (\mu-\nu)/\sigma, \tilde X < -\xi, |\tilde Y | < \xi ), \label{case21}
	\end{align}
	with equality if and only if $\sigma^2=\tau^2$. Since $(\tilde X,\tilde Y)$ has the same distribution as $( -\tilde Y,-\tilde X)$, the right hand side of \eqref{case21} is equal to
	\begin{align*}
	&P( - \tilde X < (\tau/\sigma) (-\tilde Y) + (\mu-\nu)/\sigma, - \tilde Y < -\xi, | - \tilde X | < \xi ) \\
	&= P( X > Y, \tilde Y > \xi, | \tilde X | < \xi ).
	\end{align*}
	It follows that
	\[
	\eqref{2c} +\eqref{1a} \geq P(|\tilde X| < \xi, \tilde Y > \xi),
	\]
	with equality if and only if $\sigma^2 = \tau^2$. In a similar way as in Case 1 we can deduce from this that the coverage of $I$ is greater than or equal to \eqref{toprove}, with equality if and only if $\sigma^2 = \tau^2$.

	\textbf{Case 3.} Suppose $\mu> \nu$, $\sigma^2> \tau^2$ and $\mu-\nu \leq (\sigma-\tau) \xi$. Using the fact that $(\tilde X,\tilde Y)$ has the same distribution as $(-\tilde X,-\tilde Y)$, we see that
	\begin{align*}
	\eqref{1c} &= P(\tilde X < (\tau/\sigma) \tilde Y - (\mu-\nu)/\sigma, | \tilde X | < \xi, \tilde Y < -\xi) \\
	&< P(\tilde X < (\tau/\sigma) \tilde Y + (\mu-\nu)/\sigma, | \tilde X | < \xi, \tilde Y < -\xi) \\
	&= P( -\tilde X < (\tau/\sigma) (-\tilde Y) + (\mu-\nu)/\sigma, | -\tilde X | < \xi, -\tilde Y < -\xi) \\
	&= P(X> Y, |\tilde X| < \xi, \tilde Y > \xi).
	\end{align*}
	It follows that
	\[
	\eqref{1a} + \eqref{1c} < P(|\tilde X| < \xi, \tilde Y > \xi).
	\]
	Similar to the proof of Case 1 we obtain that the coverage of $I$ is less than \eqref{toprove}.

	\textbf{Case 4.} Suppose $\mu> \nu$, $\sigma^2> \tau^2$ and $\mu-\nu \geq (\sigma-\tau) \xi$. Then the coverage of $I$ is less than \eqref{toprove}. The proof is similar to the proof of Case 2 and we omit it.
\end{proof}

In the proof of Theorem \ref{thmcoverage} we have not used the specific shape of the normal distribution. The statement of the theorem holds more generally if $(X,Y)$ has a continuous distribution with unknown mean $(\mu,\nu)$ and known variances $\sigma^2,\tau^2$ such that, for all $\mu$ and $\nu$, the distribution of $((X-\mu)/\sigma, (Y-\nu)/\tau)$ does not depend on $\mu,\nu$ and is invariant under the reflections $f_1(u,v) = (v,u)$ and $f_2(u,v) = (-v,-u)$.

\vspace{2cm}
\section{Supplemental material to simulation study}  \label{sec:appendix_B}
\vspace{2cm}
\noindent\begin{minipage}{\textwidth}
	\centering
		\includegraphics[trim=50 30 50 30, clip=true, width= \textwidth] {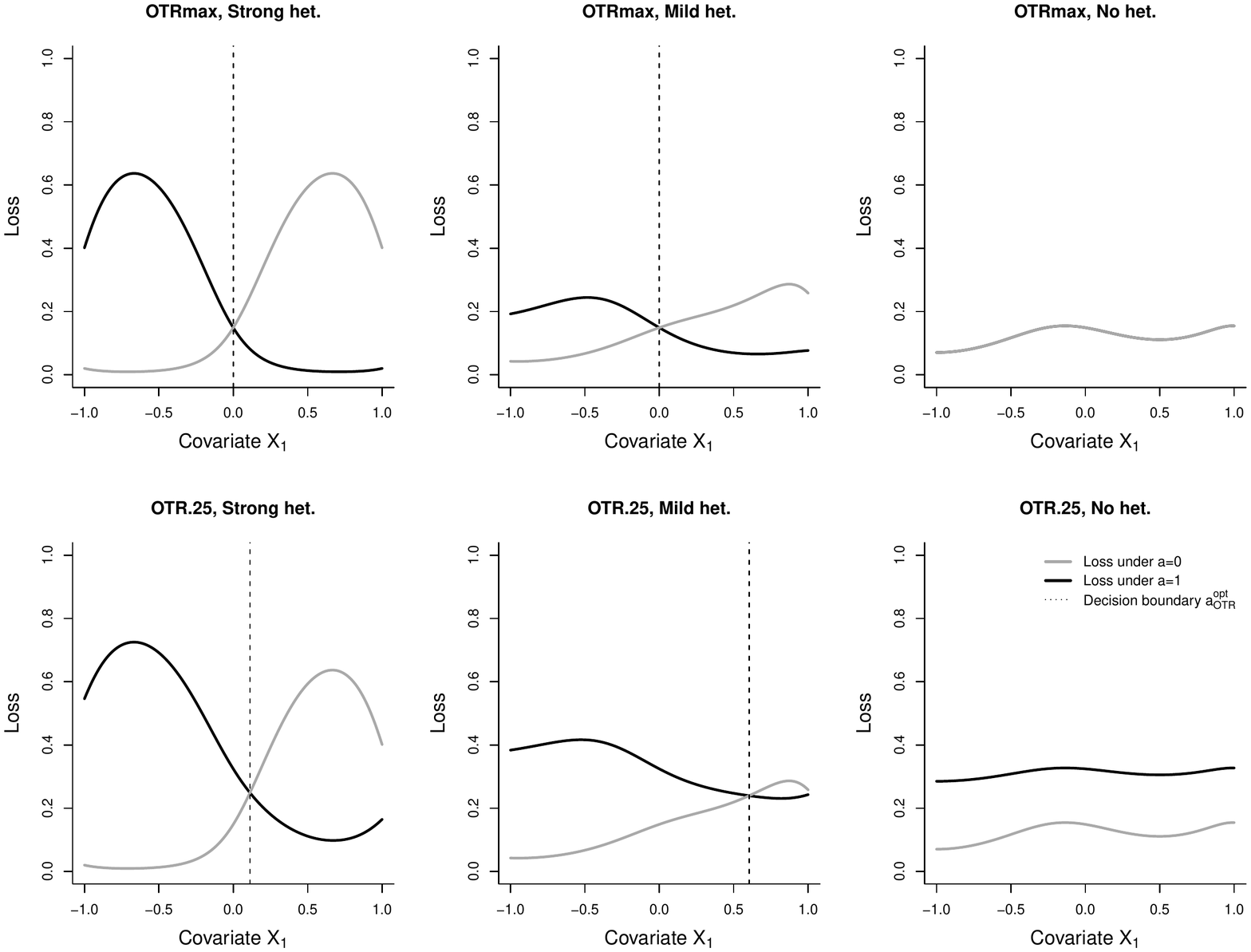}
			\captionof{figure}{True loss functions specified for the simulation by parametrizations OTRmax and OTR.25 and treatment effect heterogeneity (het.). The condition 'OTRmax, no het.' has no unique decision boundary (no treatment is optimal). The condition 'OTR.25, no het.' has no optimal decision boundary and $a=0$ is optimal for any $X_1$.} \label{fig:true_loss}
\end{minipage}
\clearpage

\begin{sidewaystable}  [!htp]
	\centering
	\caption{Simulation results for BART without noise variables ($q=0$) with OTRmax ($\phi=1$) and OTR.25 ($\phi=5$) loss functions (Het.: treatment effect heterogeneity, $\lambda$: selectivity odds ratio, $n$: sample size, $B$: Average bias of OTR posterior mean estimate, $\omega$: Average width of credible intervals, $C$: Average coverage probability of 95\% credible intervals, $A$: Accuracy of assignment.)}	
	\begin{tabular}{cccccccccccccccccc}	\label{tabA:simres_BART_q0}
		&&& \multicolumn{7}{c}{BART ($q=0$), OTRmax} && \multicolumn{7}{c}{BART ($q=0$), OTR.50} \\
		\cline{4-10}
		\cline{12-18}
		Het. &$\lambda$ &  n & $B_{L}$ & $B_Y$ & $\omega_{L}$ & $\omega_{Y}$ & $C_{L}$ & $C_Y$ & $A$ &  & $B_{L}$ & $B_Y$ & $\omega_{L}$ & $\omega_{Y}$ & $C_{L}$ & $C_Y$ & $A$ \\ \hline
		\multirow{9}{*}{\rotatebox[origin=c]{90}{Strong}} &
		\multirow{3}{*}{$-\log(3)$} & 250 & 0.007 & -0.013 & 0.191 & 0.490 & 0.982 & 0.973 & 0.958 &   & 0.007 & -0.010 & 0.195 & 0.489 & 0.979 & 0.971 & 0.955 \\
		&& 500 & 0.007 & -0.016 & 0.161 & 0.436 & 0.987 & 0.982 & 0.966 &   & 0.008 & -0.014 & 0.162 & 0.435 & 0.984 & 0.982 & 0.964 \\
		&& 1000 & 0.008 & -0.017 & 0.134 & 0.377 & 0.987 & 0.986 & 0.974 &   & 0.008 & -0.016 & 0.135 & 0.377 & 0.984 & 0.986 & 0.970 \\ \cline{2-18}
		&\multirow{3}{*}{$0$} & 250 & 0.007 & -0.016 & 0.189 & 0.449 & 0.988 & 0.986 & 0.958 &   & 0.009 & -0.017 & 0.193 & 0.451 & 0.989 & 0.987 & 0.955 \\
		&& 500 & 0.008 & -0.019 & 0.159 & 0.386 & 0.989 & 0.989 & 0.967 &   & 0.008 & -0.017 & 0.157 & 0.386 & 0.988 & 0.988 & 0.964 \\
		&& 1000 & 0.009 & -0.021 & 0.132 & 0.325 & 0.987 & 0.988 & 0.973 &   & 0.008 & -0.019 & 0.129 & 0.326 & 0.984 & 0.987 & 0.970 \\ \cline{2-18}
		&\multirow{3}{*}{$\log(3)$} & 250 & 0.012 & -0.018 & 0.204 & 0.425 & 0.985 & 0.989 & 0.960 &   & 0.012 & -0.013 & 0.206 & 0.427 & 0.985 & 0.986 & 0.951 \\
		&& 500 & 0.010 & -0.020 & 0.170 & 0.359 & 0.987 & 0.988 & 0.966 &   & 0.010 & -0.018 & 0.168 & 0.362 & 0.985 & 0.987 & 0.963 \\
		&& 1000 & 0.010 & -0.021 & 0.142 & 0.299 & 0.987 & 0.986 & 0.972 &   & 0.010 & -0.019 & 0.137 & 0.301 & 0.985 & 0.985 & 0.969 \\ \hline
		\multirow{9}{*}{\rotatebox[origin=c]{90}{Mild}} & 
		\multirow{3}{*}{$-\log(3)$} & 250 & -0.022 & 0.033 & 0.288 & 0.502 & 0.985 & 0.980 & 0.816 &   & -0.010 & 0.030 & 0.295 & 0.498 & 0.986 & 0.982 & 0.842 \\
		&& 500 & -0.012 & 0.021 & 0.250 & 0.446 & 0.988 & 0.988 & 0.853 &   & -0.007 & 0.019 & 0.256 & 0.436 & 0.99 & 0.988 & 0.850 \\
		&& 1000 & -0.005 & 0.012 & 0.213 & 0.386 & 0.991 & 0.991 & 0.885 &   & -0.005 & 0.013 & 0.219 & 0.373 & 0.991 & 0.991 & 0.854 \\ \cline{2-18}
		&\multirow{3}{*}{$0$} & 250 & -0.017 & 0.026 & 0.288 & 0.467 & 0.989 & 0.988 & 0.817 &   & -0.006 & 0.019 & 0.303 & 0.474 & 0.990 & 0.989 & 0.821 \\
		&& 500 & -0.009 & 0.015 & 0.245 & 0.399 & 0.992 & 0.990 & 0.864 &   & -0.006 & 0.015 & 0.257 & 0.407 & 0.991 & 0.993 & 0.835 \\
		&& 1000 & -0.003 & 0.008 & 0.207 & 0.336 & 0.992 & 0.990 & 0.887 &   & -0.004 & 0.010 & 0.216 & 0.343 & 0.993 & 0.993 & 0.847 \\ \cline{2-18}
		&\multirow{3}{*}{$\log(3)$} & 250 & -0.021 & 0.035 & 0.300 & 0.452 & 0.982 & 0.979 & 0.790 &   & -0.007 & 0.026 & 0.322 & 0.470 & 0.984 & 0.983 & 0.799 \\
		&& 500 & -0.010 & 0.019 & 0.261 & 0.381 & 0.991 & 0.986 & 0.839 &   & -0.005 & 0.019 & 0.279 & 0.406 & 0.990 & 0.989 & 0.821 \\
		&& 1000 & -0.004 & 0.011 & 0.220 & 0.314 & 0.992 & 0.985 & 0.876 &   & -0.005 & 0.013 & 0.236 & 0.344 & 0.991 & 0.991 & 0.837 \\ \hline
		\multirow{9}{*}{\rotatebox[origin=c]{90}{None}} &
		\multirow{3}{*}{$-\log(3)$} & 250 & -0.036 & 0.054 & 0.300 & 0.448 & 0.983 & 0.977 & 0.519 &   & 0.001 & 0.025 & 0.302 & 0.449 & 0.984 & 0.977 & 0.930 \\
		&& 500 & -0.030 & 0.042 & 0.264 & 0.386 & 0.99 & 0.983 & 0.506 &   & 0.006 & 0.014 & 0.262 & 0.385 & 0.988 & 0.982 & 0.959 \\
		&& 1000 & -0.023 & 0.033 & 0.228 & 0.328 & 0.992 & 0.986 & 0.503 &   & 0.009 & 0.004 & 0.223 & 0.323 & 0.99 & 0.986 & 0.979 \\ \cline{2-18}
		&\multirow{3}{*}{$0$} & 250 & -0.035 & 0.049 & 0.295 & 0.435 & 0.988 & 0.984 & 0.501 &   & 0.006 & 0.014 & 0.312 & 0.434 & 0.989 & 0.984 & 0.939 \\
		&& 500 & -0.027 & 0.037 & 0.257 & 0.372 & 0.992 & 0.988 & 0.495 &   & 0.010 & 0.004 & 0.270 & 0.371 & 0.990 & 0.987 & 0.965 \\
		&& 1000 & -0.022 & 0.031 & 0.220 & 0.313 & 0.993 & 0.987 & 0.496 &   & 0.011 & -0.001 & 0.228 & 0.311 & 0.992 & 0.987 & 0.986 \\ \cline{2-18}
		&\multirow{3}{*}{$\log(3)$} & 250 & -0.036 & 0.055 & 0.300 & 0.448 & 0.984 & 0.976 & 0.476 &   & 0.006 & 0.018 & 0.338 & 0.444 & 0.983 & 0.977 & 0.913 \\
		&& 500 & -0.029 & 0.042 & 0.265 & 0.387 & 0.988 & 0.983 & 0.491 &   & 0.01 & 0.009 & 0.294 & 0.385 & 0.989 & 0.985 & 0.957 \\
		&& 1000 & -0.024 & 0.033 & 0.228 & 0.328 & 0.992 & 0.985 & 0.493 &   & 0.012 & 0.002 & 0.253 & 0.327 & 0.991 & 0.985 & 0.978 \\ \hline
	\end{tabular}
\end{sidewaystable}

\clearpage

\begin{sidewaystable}  [p]
	\centering
	\caption{Simulation results for the Bayesian logistic regression model with noise variables ($q=5$) with OTRmax ($\phi=1$) and OTR.25 ($\phi=5$) loss functions (Het.: treatment effect heterogeneity, $\lambda$: selectivity odds ratio, $n$: sample size, $B$: Average bias of OTR posterior mean estimate, $\omega$: Average width of credible intervals, $C$: Average coverage probability of 95\% credible intervals, $A$: Accuracy of assignment.)}	
	\begin{tabular}{cccccccccccccccccc}	\label{tabA:simres_mcmclogit_q5}
		&&& \multicolumn{7}{c}{Logistic regression ($q=5$), OTRmax} && \multicolumn{7}{c}{Logistic regression ($q=5$), OTR.50} \\
		\cline{4-10}
		\cline{12-18}
		Het. &$\lambda$ &  n & $B_{L}$ & $B_Y$ & $\omega_{L}$ & $\omega_{Y}$ & $C_{L}$ & $C_Y$ & $A$ &  & $B_{L}$ & $B_Y$ & $\omega_{L}$ & $\omega_{Y}$ & $C_{L}$ & $C_Y$ & $A$ \\ \hline
		\multirow{9}{*}{\rotatebox[origin=c]{90}{Strong}} &
		\multirow{3}{*}{$-\log(3)$} & 250 & -0.045 & 0.040 & 0.208 & 0.570 & 0.726 & 0.830 & 0.892 &   & -0.037 & 0.044 & 0.252 & 0.565 & 0.805 & 0.824 & 0.878 \\
		&& 500 & -0.018 & 0.018 & 0.161 & 0.458 & 0.840 & 0.881 & 0.938 &   & -0.014 & 0.019 & 0.179 & 0.456 & 0.876 & 0.883 & 0.926 \\
		&& 1000 & -0.008 & 0.007 & 0.118 & 0.342 & 0.886 & 0.901 & 0.959 &   & -0.006 & 0.010 & 0.125 & 0.344 & 0.903 & 0.903 & 0.954 \\  \cline{2-18}
		&\multirow{3}{*}{$0$} & 250 & -0.041 & 0.039 & 0.205 & 0.540 & 0.741 & 0.839 & 0.898 &   & -0.033 & 0.044 & 0.242 & 0.539 & 0.818 & 0.839 & 0.888 \\
		&& 500 & -0.016 & 0.017 & 0.157 & 0.418 & 0.848 & 0.887 & 0.940 &   & -0.012 & 0.018 & 0.170 & 0.421 & 0.878 & 0.888 & 0.932 \\
		&& 1000 & -0.007 & 0.008 & 0.114 & 0.309 & 0.889 & 0.906 & 0.960 &   & -0.005 & 0.008 & 0.117 & 0.310 & 0.904 & 0.906 & 0.957 \\  \cline{2-18}
		&\multirow{3}{*}{$\log(3)$} & 250 & -0.040 & 0.042 & 0.208 & 0.523 & 0.732 & 0.844 & 0.900 &   & -0.032 & 0.044 & 0.241 & 0.525 & 0.809 & 0.844 & 0.891 \\
		&& 500 & -0.016 & 0.016 & 0.161 & 0.400 & 0.843 & 0.887 & 0.94 &   & -0.012 & 0.019 & 0.173 & 0.404 & 0.875 & 0.889 & 0.932 \\
		&& 1000 & -0.007 & 0.007 & 0.117 & 0.294 & 0.887 & 0.906 & 0.962 &   & -0.005 & 0.008 & 0.119 & 0.297 & 0.9 & 0.907 & 0.957 \\  \hline
		\multirow{9}{*}{\rotatebox[origin=c]{90}{Mild}} &
		\multirow{3}{*}{$-\log(3)$} & 250 & -0.095 & 0.096 & 0.299 & 0.577 & 0.746 & 0.850 & 0.693 &   & -0.058 & 0.082 & 0.336 & 0.573 & 0.812 & 0.847 & 0.740 \\
		&& 500 & -0.057 & 0.056 & 0.250 & 0.456 & 0.839 & 0.892 & 0.760 &   & -0.034 & 0.047 & 0.263 & 0.450 & 0.876 & 0.883 & 0.787 \\
		&& 1000 & -0.031 & 0.031 & 0.192 & 0.344 & 0.881 & 0.907 & 0.829 &   & -0.019 & 0.026 & 0.198 & 0.334 & 0.903 & 0.899 & 0.830 \\  \cline{2-18}
		&\multirow{3}{*}{$0$} & 250 & -0.091 & 0.090 & 0.295 & 0.556 & 0.755 & 0.854 & 0.702 &   & -0.057 & 0.077 & 0.335 & 0.555 & 0.817 & 0.850 & 0.736 \\
		&& 500 & -0.052 & 0.053 & 0.245 & 0.430 & 0.850 & 0.893 & 0.773 &   & -0.032 & 0.044 & 0.260 & 0.431 & 0.880 & 0.890 & 0.789 \\
		&& 1000 & -0.028 & 0.028 & 0.186 & 0.316 & 0.886 & 0.906 & 0.841 &   & -0.019 & 0.025 & 0.194 & 0.317 & 0.906 & 0.905 & 0.829 \\  \cline{2-18}
		&\multirow{3}{*}{$\log(3)$} & 250 & -0.093 & 0.096 & 0.300 & 0.553 & 0.751 & 0.850 & 0.696 &   & -0.059 & 0.083 & 0.348 & 0.561 & 0.812 & 0.857 & 0.731 \\
		&& 500 & -0.055 & 0.057 & 0.251 & 0.426 & 0.844 & 0.89 & 0.763 &   & -0.033 & 0.046 & 0.273 & 0.433 & 0.879 & 0.89 & 0.785 \\
		&& 1000 & -0.030 & 0.030 & 0.191 & 0.310 & 0.882 & 0.905 & 0.838 &   & -0.020 & 0.026 & 0.202 & 0.319 & 0.899 & 0.903 & 0.826 \\  \hline
		\multirow{9}{*}{\rotatebox[origin=c]{90}{None}} &
		\multirow{3}{*}{$-\log(3)$} & 250 & -0.101 & 0.105 & 0.285 & 0.522 & 0.705 & 0.832 & 0.484 &   & -0.055 & 0.078 & 0.333 & 0.509 & 0.783 & 0.814 & 0.782 \\
		&& 500 & -0.071 & 0.072 & 0.249 & 0.408 & 0.811 & 0.881 & 0.485 &   & -0.026 & 0.042 & 0.271 & 0.402 & 0.869 & 0.873 & 0.863 \\
		&& 1000 & -0.049 & 0.050 & 0.198 & 0.302 & 0.864 & 0.905 & 0.492 &   & -0.008 & 0.017 & 0.207 & 0.296 & 0.905 & 0.902 & 0.936 \\ \cline{2-18}
		&\multirow{3}{*}{$0$} & 250 & -0.099 & 0.101 & 0.284 & 0.511 & 0.712 & 0.834 & 0.501 &   & -0.053 & 0.077 & 0.34 & 0.511 & 0.796 & 0.830 & 0.789 \\
		&& 500 & -0.068 & 0.068 & 0.245 & 0.399 & 0.821 & 0.89 & 0.505 &   & -0.025 & 0.039 & 0.274 & 0.396 & 0.874 & 0.882 & 0.865 \\
		&& 1000 & -0.047 & 0.047 & 0.192 & 0.291 & 0.864 & 0.907 & 0.502 &   & -0.008 & 0.017 & 0.207 & 0.292 & 0.902 & 0.903 & 0.937 \\  \cline{2-18}
		&\multirow{3}{*}{$\log(3)$} & 250 & -0.101 & 0.104 & 0.285 & 0.522 & 0.705 & 0.831 & 0.515 &   & -0.058 & 0.083 & 0.351 & 0.520 & 0.778 & 0.822 & 0.78 \\
		&& 500 & -0.070 & 0.071 & 0.250 & 0.410 & 0.817 & 0.886 & 0.508 &   & -0.030 & 0.045 & 0.286 & 0.411 & 0.859 & 0.88 & 0.853 \\
		&& 1000 & -0.049 & 0.049 & 0.199 & 0.303 & 0.865 & 0.908 & 0.510 &   & -0.010 & 0.019 & 0.222 & 0.308 & 0.902 & 0.908 & 0.926 \\ \hline
	\end{tabular}
\end{sidewaystable}	

\clearpage
\newpage

\section{Supplemental material to oropharynx OTR study}  \label{sec:appendix_C}
\begin{figure}[!ht]  
	\begin{subfigure}{0.32\textwidth}
		\includegraphics[trim=140 0 150 0, clip=true, width=\linewidth]{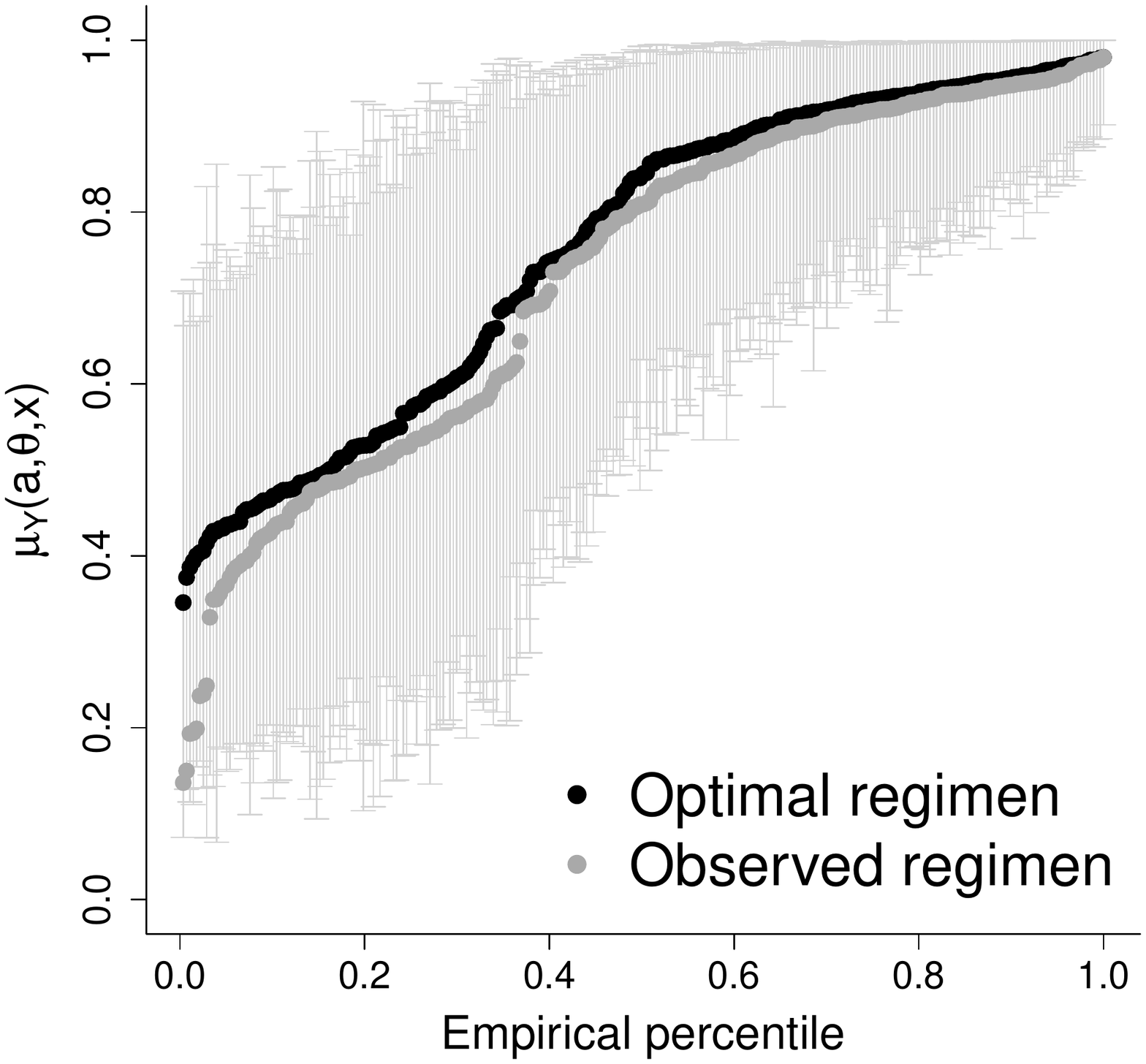}
		\caption{OTRmax} 
	\end{subfigure}  
	\begin{subfigure}{0.32\textwidth}
		\includegraphics[trim=140 0 150 0, clip=true, width=\linewidth]{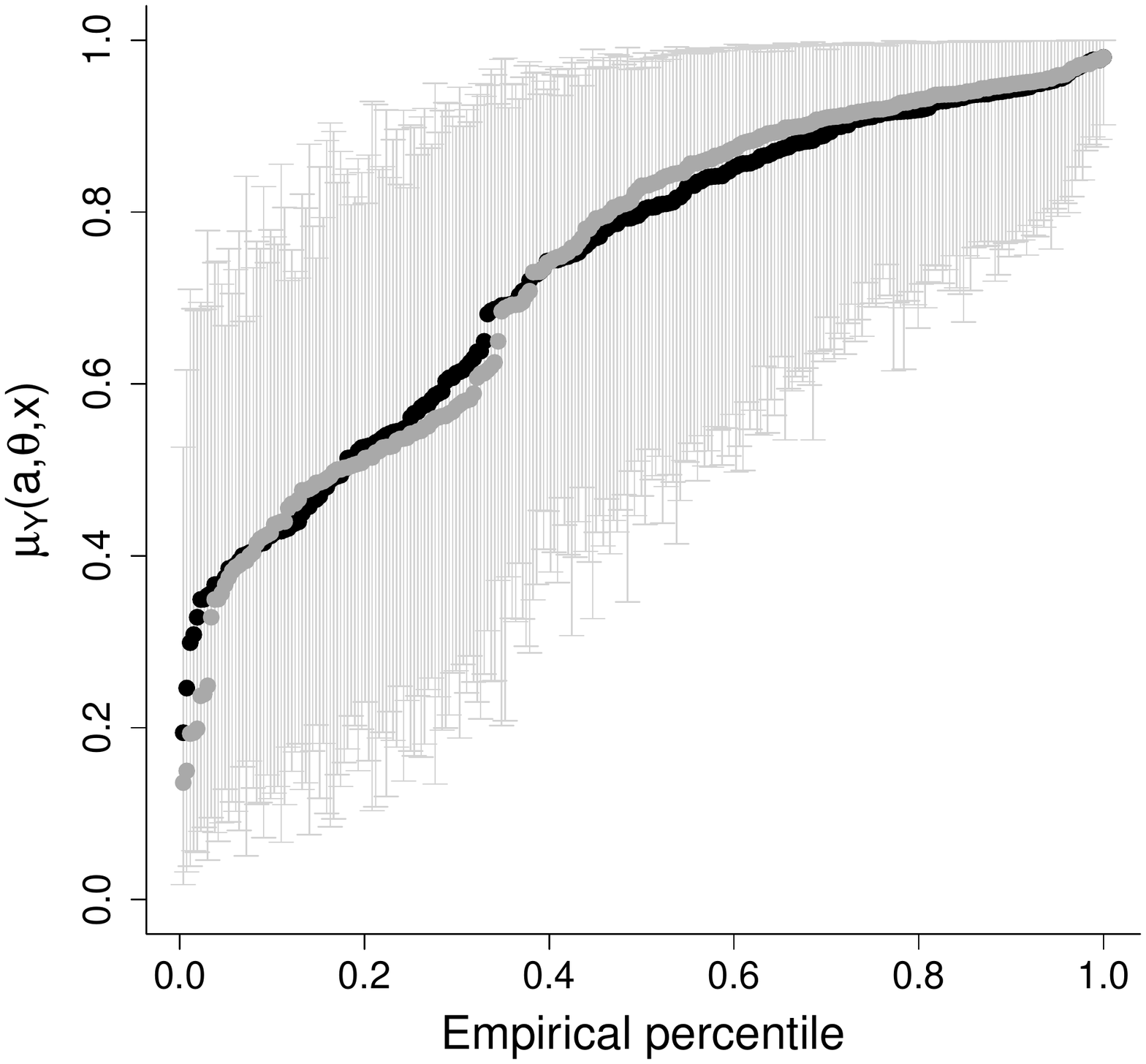}
		\caption{OTR.25} 
	\end{subfigure}
	\begin{subfigure}{0.32\textwidth}
		\includegraphics[trim=140 0 150 0, clip=true, width=\linewidth]{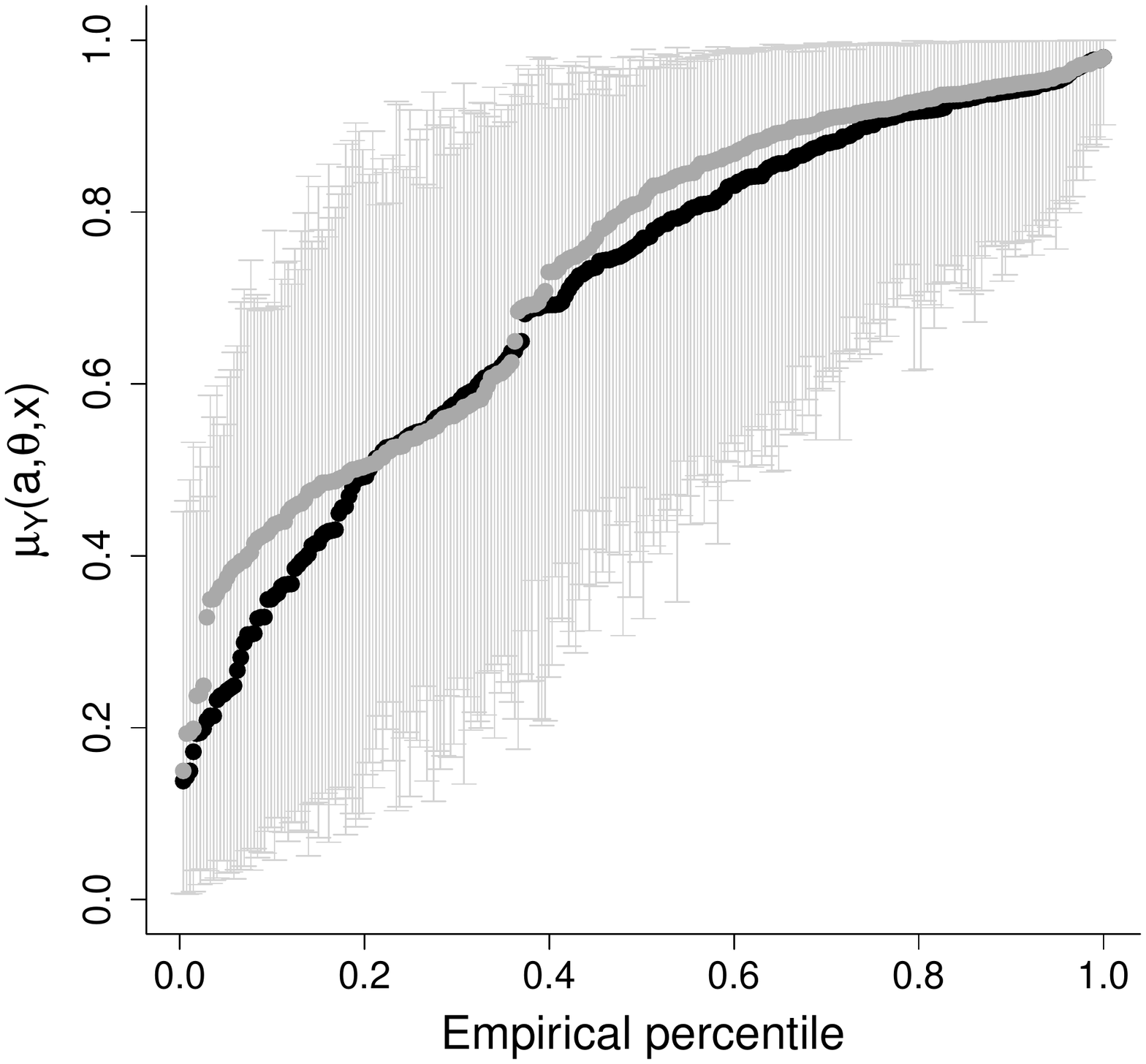}
		\caption{OTR.50} 
	\end{subfigure}
	\caption{Mean posterior survival probabilities by OTR $\mu_Y(a_1^*(\x_i),\theta,\x_i)$ as compared to the observed regime $\mu_Y(w_i,\theta,\x_i)$. 95\% credible intervals given for $\mu_Y(a_1^*(\x_i),\theta,\x_i)$. Patients with non-robust decisions omitted. }
\end{figure}

\begin{figure}[!ht]  
	\begin{subfigure}{0.32\textwidth}
		\includegraphics[trim=140 0 150 0, clip=true, width=\linewidth]{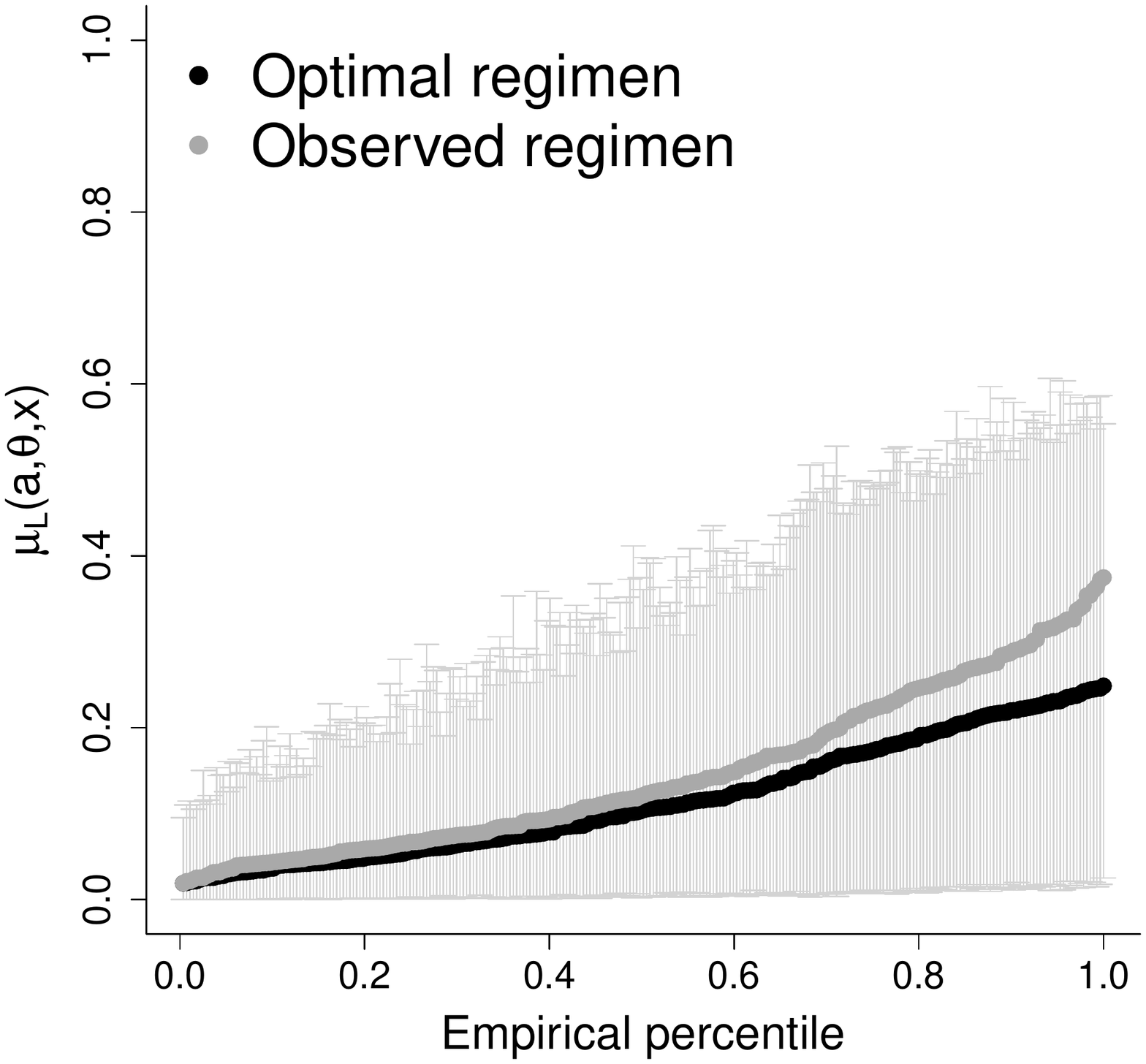}
		\caption{OTRmax} 
	\end{subfigure} 
	\begin{subfigure}{0.32\textwidth}
		\includegraphics[trim=140 0 150 0, clip=true, width=\linewidth]{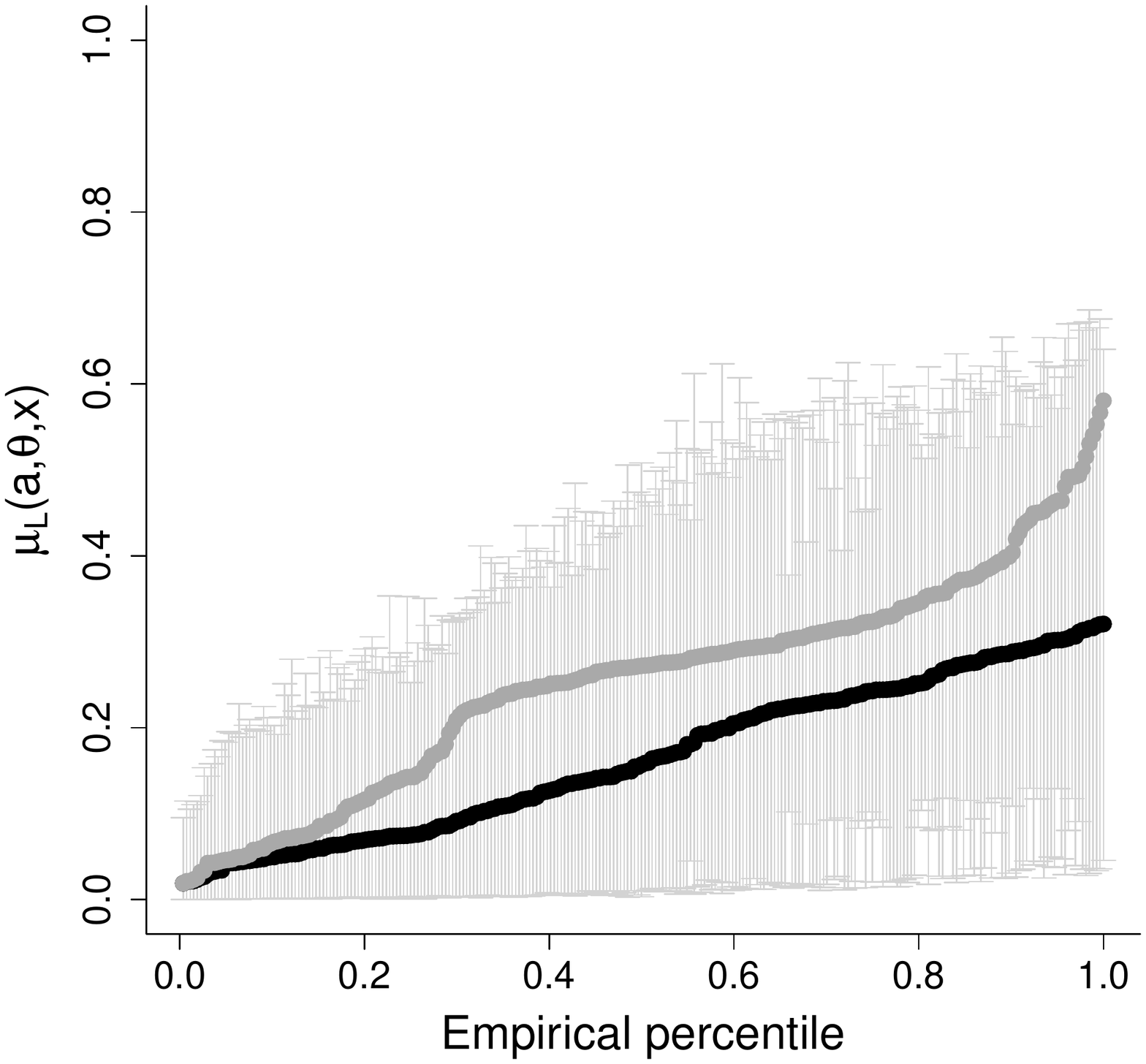}
		\caption{OTR.25} 
	\end{subfigure}
	\begin{subfigure}{0.32\textwidth}
		\includegraphics[trim=140 0 150 0, clip=true, width=\linewidth]{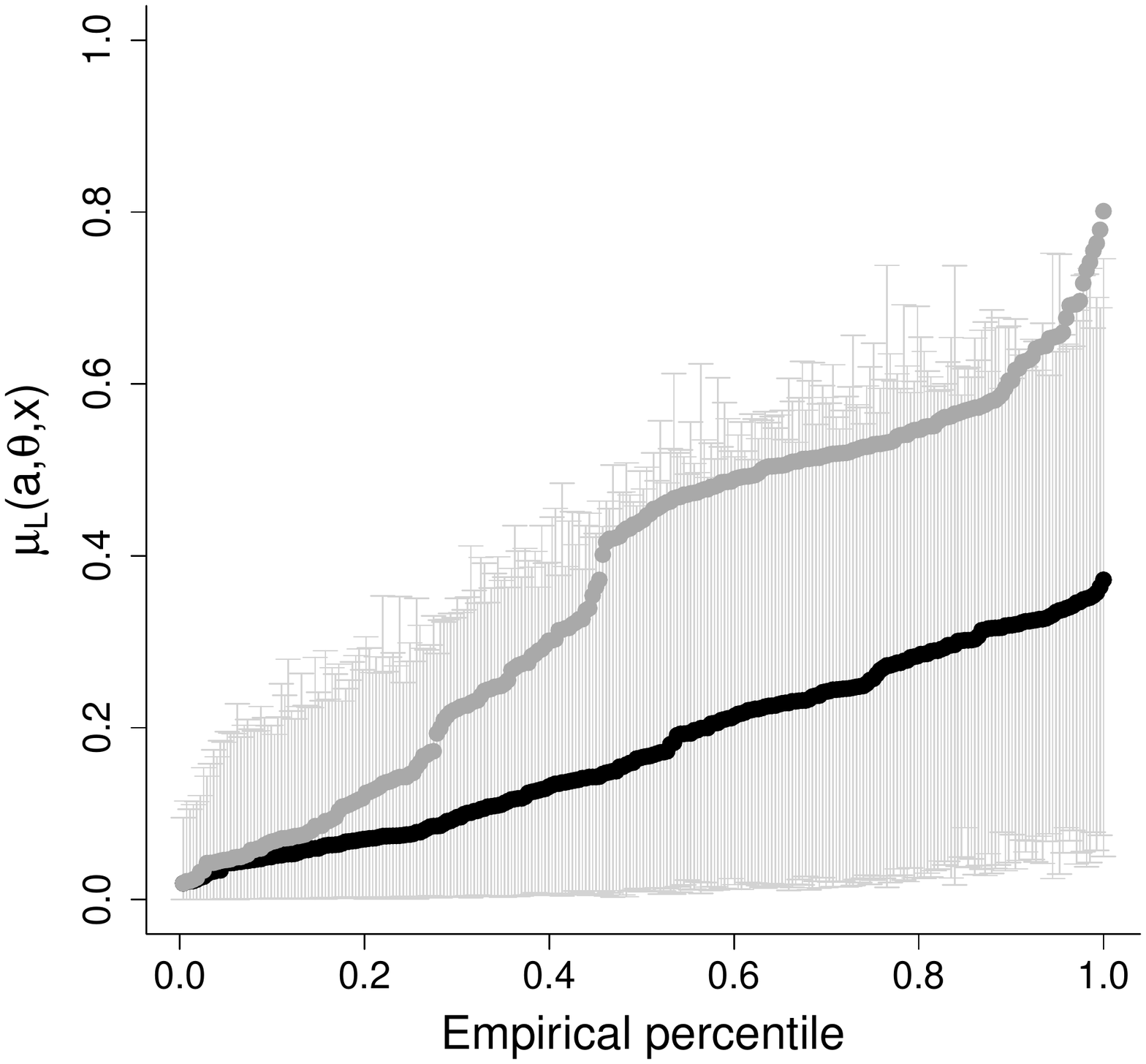}
		\caption{OTR.50} 
	\end{subfigure}
	\caption{Mean posterior expected loss $\mu_{\Loss}(a_1^*(\x_i),\theta,\x_i)$ under OTR (optimal regime) and the observed regime $\mu_{\Loss}(w_i,\theta,\x_i)$. Credible intervals given for OTR. Parameter uncertainty of $\phi$ simulated by prior $\phi \sim \text{U}(\exp(-3), \exp(3))$. Patients with non-robust decisions omitted. }
\end{figure}

\clearpage

\begin{center} 
	\begin{figure} 
	\scalebox{.7}{	\includegraphics[trim=130 30 130 30, clip=true, width=\linewidth]{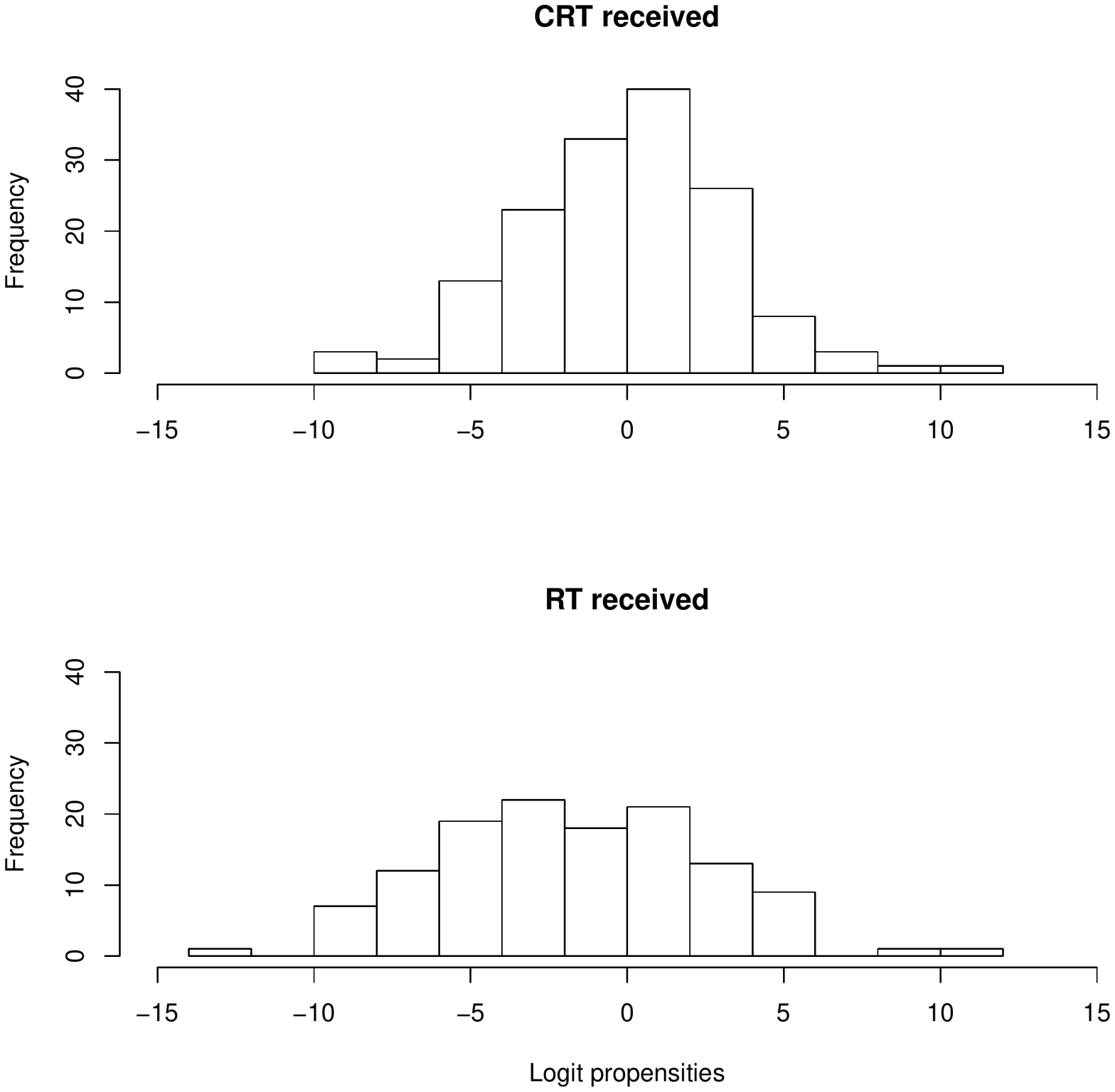}} 
	\caption{Logit propensities by treatments received. Propensities were estimated using a logistic treatment assignment model. The distributions demonstrate overlap in support of the overlap assumption under the Rubin Causal Model.}
		\label{fig:selectivity_analysis}
	\end{figure}
\end{center}

\end{appendices}
\end{document}